\documentclass{article}

\usepackage[utf8]{inputenc}

\usepackage[T1]{fontenc}

\usepackage{times}  
\usepackage{helvet}  
\usepackage{courier}  
\usepackage{url}  
\usepackage{graphicx}  
\usepackage{paralist}
\usepackage{amsmath}
\usepackage{amssymb}
\usepackage{amsthm}
\usepackage{nicefrac}
\usepackage{tikz}
\usetikzlibrary{backgrounds,positioning,calc}
\usepackage{mdframed}
\usepackage{pgfplots}
\usepackage{booktabs}

\usepackage{xcolor}
\usepackage{hyperref}
\hypersetup{
    colorlinks=true,
    linkcolor=olive,
    citecolor=purple
}

\usepackage[textwidth=35mm]{todonotes}
\usepackage{xspace}

\allowdisplaybreaks

\newcommand{\pref}{\succ}

\newcommand{\p}{{\mathrm{P}}}
\newcommand{\np}{{\mathrm{NP}}}

\newcommand{\fpt}{{\mathrm{FPT}}}

\newcommand{\calA}{\mathcal{A}}

\newcommand{\calD}{\mathcal{D}}
\newcommand{\calL}{\mathcal{L}}

\newcommand{\discrete}{{{\mathrm{disc}}}}

\newcommand{\swap}{{{\mathrm{swap}}}}
\newcommand{\spearman}{{{\mathrm{Spear}}}}

\newcommand{\cost}{{w}}
\newcommand{\pos}{{{{\mathrm{pos}}}}}

\DeclareMathOperator*{\opt}{\textsc{OPT}}
\DeclareMathOperator*{\sol}{\textsc{SOL}}
\DeclareMathOperator*{\poly}{poly}

\newcommand{\mybrace}[2]{\overbrace{\rule{0cm}{0.5cm}#1}^{#2}}

\newcommand{\id}[1]{{#1\hbox{-}\mathrm{ID}}}

\def\2vec#1#2{\left(\begin{array}{c}{#1}\\{#2}\end{array}\right)}

\newtheorem{definition}{Definition}
\newtheorem{observation}{Observation}
\newtheorem{theorem}{Theorem}

\newtheorem{proposition}[theorem]{Proposition}
\newtheorem{corollary}{Corollary}

\newtheorem{example}{Example}

\newcommand{\cutforAAAI}[1]{}

\begin{document}

\title{How Similar Are Two Elections?\thanks{Early version of this 
paper was presented at AAAI-2019~\cite{FaliszewskiSSST19}. This version includes additional results and revised discussions, taking into account follow-up works and the progress since that time.}}
\author{Piotr Faliszewski \\
  AGH University\\
  Krak\'{o}w, Poland 
\and
  Piotr Skowron\\
  University of Warsaw \\
  Warsaw, Poland
\and
  Arkadii Slinko\\
  University of Auckland\\
  Auckland, New Zealand
\and
  Krzysztof Sornat\\
  AGH University\\
  Krak\'{o}w, Poland 
\and
  Stanis{\l}aw Szufa\\
  AGH University\\
  Krak\'{o}w, Poland
\and
  Nimrod Talmon\\
  Ben-Gurion University\\
  Be'er Sheva, Israel
}
\date{}

\maketitle

\begin{abstract}
  We introduce and study isomorphic distances between ordinal
  elections (with the same numbers of candidates and voters). The main
  feature of these distances is that they are invariant to renaming
  the candidates and voters, and two elections are at distance zero if
  and only if they are isomorphic. Specifically, we consider
  isomorphic extensions of distances between preference orders: Given
  such a distance $d$, we extend it to distance $\id{d}$ between
  elections by unifying candidate names and finding a matching between
  the votes, so that the sum of the $d$-distances between the matched
  votes is as small as possible.
  
  We show that testing isomorphism of two elections can be done in
  polynomial time so, in principle, such distances can be tractable.
  Yet, we show that two very natural isomorphic distances are
  $\np$-complete and hard to approximate. We attempt to rectify the
  situation by showing $\fpt$ algorithms for several natural
  parameterizations.

\end{abstract}

\section{Introduction}

We consider the ordinal model of elections, where each voter submits a
preference order that ranks the candidates from the most to the least
desirable one.  Given two such elections of equal size---i.e., with
the same numbers of candidates and the same numbers of voters, albeit
where the names of the candidates and voters may differ---we want to
know how structurally similar they are.

To this end, we design distances that are invariant to renaming the
candidates and voters, and which ensure that two elections are at
distance zero if and only if they are isomorphic.  We study the
complexity of computing several such distances and seek ways of
circumventing their intractability.

Our starting point is the distance rationalizability
framework~\cite{nit:j:closeness,mes-nur:b:distance-realizability,elk-fal-sli:j:dr,elk-fal-sli:j:distance-rational,elk-sli:b:rationalization},
which also considers distances over elections, but which assumes
identical candidate and voter sets (for example, one can think of
elections with voters' preferences at different points of time).

One approach from this framework is to take some metric $d$ over preference orders (such as
the swap distance, which counts the number of inversions) and extend
it to elections by summing the distances between the orders submitted
by each voter in both elections.

We adapt this idea, taking into account that the candidate and voter
sets of the input elections may be different (albeit, we crucially
require equal numbers of candidates and equal numbers of voters):

We first rename the candidates in one of the elections to be the same
as in the other one, then we match each voter from one election to a
distinct voter in the other one, and finally we sum up the distances
between the preference orders of the matched voters.  Importantly, we
rename the candidates and match the voters in such a way as to
minimize the final outcome. We refer to the thus-defined metric as the
$d$-\textsc{Isomorphic Distance} ($\id{d}$ for short). Such distances
indeed are invariant to renaming the candidates and voters, and ensure
that elections are at distance zero exactly if they are isomorphic.

Yet, it is natural to worry about the complexity of computing such
distances because, irrespective of the choice of $d$, being able to
compute $\id{d}$ implies the ability to decide if two elections are
isomorphic.  Fortunately, even though the complexity of testing
isomorphism of many mathematical objects is elusive (where the case of
\textsc{Graph Isomorphism} is by far the most famous example; see, e.g.,
the report of Babai et al.~\cite{bab-daw-sch-tor:j:graphi-isomorphism}
and further discussion on Babai's home page (2017)), ordinal elections are so structured that for them
testing isomorphism is easy.

Sadly, this is where most good news end.  For example, if we take $d$
to be the swap distance, then $d$-\textsc{Isomorphic Distance}
generalizes the {\sc Kemeny Score} problem, which is known to
be computationally
hard~\cite{bar-tov-tri:j:who-won,dwo-kum-nao-siv:c:rank-aggregation}. Thus
we consider several approaches to circumventing the hardness of
computing our distances.

\subsection{Motivation}
Election isomorphism and isomorphic distances seem to be both
fundamental for understanding elections, and quite useful.

For example, we use the isomorphism idea 
to pinpoint an important difference between the single-peaked and the
single-crossing domains (see Section~\ref{sec:ei}).

Further, our isomorphic distances allow for a principled way of
performing classification, clustering, and other common tasks related
to distances, over the space of elections.

Below, we give three more specific examples of applications:

\begin{description}
\item[Map of Elections.]  Following up on our work, Szufa et
  al.~\cite{szu-fal-sko-sli-tal:c:map} and Boehmer et
  al.~\cite{boe-bre-fal-nie-szu:c:compass} introduced the
  map-of-elections framework. The idea is to collect a dataset of
  elections from as many different sources as possible (hence, these
  elections regard different candidates and different voters), compute
  distances between each pair using some metric, and then represent
  these elections as points on a 2D plane, so that the Euclidean
  distances would be as similar as possible to those given by the
  metric. This 2D representation is the \emph{map}. Szufa et
  al.~\cite{szu-fal-sko-sli-tal:c:map} and Boehmer et
  al.~\cite{boe-bre-fal-nie-szu:c:compass} have shown a number of ways
  in which such maps can be helpful in designing computational
  experiments and visualizing their results.  While Szufa et
  al.~\cite{szu-fal-sko-sli-tal:c:map} and Boehmer et
  al.~\cite{boe-bre-fal-nie-szu:c:compass} did not use our distances
  due to their high computational complexity (they designed simpler
  ones), maps based on isomorphic distances also exist and have
  interesting
  applications~\cite{boe-fal-nie-szu-was:c:metrics,fal-kac-sor-szu-was:c:microscope}.
  Indeed, whenever they can be computed, they are viewed as superior.

\item[Analyzing Real-Life Elections.] Given an election (e.g., from
  PrefLib~\cite{mat-wal:c:preflib}, a library of real-life elections), it is interesting to ask if it is similar to those that
  come from some distribution (such as, e.g., the impartial culture
  model, where each preference order is equally likely; 

  or one of the Euclidean models, where candidates and voters are
  mapped to points in some Euclidean space, and the voters rank the
  candidates with respect to the distance from their
  points~\cite{enelow1984spatial,enelow1990advances}). To this end, we
  may generate a number of elections according to a given distribution
  and compute their isomorphic distances from the given one.  This
  approach was taken, e.g., by Elkind et
  al.~\cite{boe-bre-elk-fal-szu:c:frequency-matrices}, albeit using
  the relaxed distances of Szufa et
  al.~\cite{szu-fal-sko-sli-tal:c:map} and Boehmer et
  al.~\cite{boe-bre-fal-nie-szu:c:compass}.

\item[Comparing Election Distributions.]  Given several samples from
  two distributions over elections, we may ask how similar these
  distributions are to each other. For example, we may wonder how
  different are two models of elections with Euclidean preferences
  where in one model voters and candidates are distributed uniformly
  on a disc and in another they are distributed

  according to a Gaussian distribution.\footnote{Such a comparison
    would be useful, e.g., in the work of Elkind et
    al.~\cite{elk-fal-las-sko-sli-tal:c:2d-multiwinner}, where
    the authors evaluated multiwinner voting rules on four Euclidean
    models of elections.} While comparing real-valued distributions is
  a classic topic within statistics, doing the same for elections
  seems much harder, and we are not aware of any good solution. Yet, intuitively, the ability to measure the isomorphic distance
  between samples of elections would be helpful in designing an
  appropriate method (indeed, e.g., the classic $\chi^2$ test relies on
  comparing frequencies of \emph{similar} items in both distributions).

\end{description}

We stress that we do not address the above problems directly, but,
rather, we claim that isomorphic distances are
important tools for tackling them. 
Our focus is on establishing
the complexity of computing our distances.

\subsection{Our Contribution}
Our main contributions are as
follows:

\begin{enumerate}
\item We show that it is possible to test election isomorphism in
  polynomial time and, using similar ideas, we show a polynomial-time
  algorithm for the \textsc{Isomorphic Distance} problem based on the
  discrete distance (but we do not expect this variant of the problem
  to be very relevant in practice).

\item We show that computing isomorphic distances based on the swap
  and Spearman metrics is $\np$-complete (see Section~\ref{sec:prelim}
  for definitions). In the former case, we inherit hardness from the
  {\sc Kemeny Score} problem, but in the latter we use a new proof. We strengthen
  these results by showing inapproximability results, but we also
  provide parametrized algorithms for both problems.

\end{enumerate}

\subsection{Related Work}

Our work was inspired by the distance rationalizability framework,
which offers a unified langauge for defining voting rules. The idea is
to consider the space of all possible elections over a given set of
candidates and with a given set of voters (who may have different
preference orders in different elections), label some of the elections
with ``obviously winning'' candidates (e.g., an election where all
voters rank a given candidate $c$ on the top position could be labeled
with $c$), take a distance $d$ over elections and define a voting rule
as follows: Candidate $c$ is a (possibly tied) winner of election $E$
if among all the labeled elections that are at the smallest
$d$-distance from $E$, there is one labeled with $c$.  By varying the
labeling of the space and the distance used, one can recover many
well-known rules, as well as rapidly develop new ones with particular
properties.

For an overview
of the distance rationalizability framework, we point the reader to the chapter of Elkind and
Slinko~\cite{elk-sli:b:rationalization}. Early work regarding the
framework is due to Nitzan~\cite{nit:j:closeness} and Meskanen and
Nurmi~\cite{mes-nur:b:distance-realizability}, whereas our isomorphic
distances are related to the votewise distances studied by Elkind,
Faliszewski, and Slinko~\cite{elk-fal-sli:j:dr}.

The idea of isomorphic distances, first presented in the conference
version of this paper in 2019, lead to the development of the
map-of-elections
framework~\cite{szu-fal-sko-sli-tal:c:map,boe-bre-fal-nie-szu:c:compass}.
In particular, in the few years between the publication of the
conference version of this paper and now, quite a number of papers
studying the map-of-elections framework were published (often
coauthored by some of the authors of this work), many of which were
either inspired by our work, or followed-up on it. Instead of listing
these papers here, we mention them at relevant points throughout the paper.

Our work was also extended by Redko et
al.~\cite{red-vay-fla-cou:c:co-optimal-transport}, who proposed an
approach to measuring distances that is not restricted to elections of
the same sizes.

Similarly, following up on the idea of comparing elections of
different sizes, Faliszewski et
al.~\cite{fal-sor-szu:c:subelection-isomorphism} considered the
complexity of testing if a given election is isomorphic to a
subelection of another one, and of deciding if two elections have
isomorphic subelections of a given size.

Finally, our work was inspired by the studies of the \textsc{Graph
  Isomorphism} problem (see, e.g., the report of Babai et
al.~\cite{bab-daw-sch-tor:j:graphi-isomorphism}).  In particular, our
isomorphic distances are related to the notion of approximate graph
isomorphism, introduced by Arvind et
al.~\cite{arv-koe-kur-vas:c:approximate-graph-isomorphism} and then
studied by Grohe et
al.~\cite{gro-rat-woe:c:approximate-graph-isomorphism}.  Approximate
graph isomorphism is very similar to our problems in spirit, but is
quite different on the technical level.

\section{Preliminaries}\label{sec:prelim}

For an integer $n$, we write $[n]$ to denote the set
$\{1, \ldots, n\}$.  For two sets $A$ and $B$ of the same cardinality,
by $\Pi(A,B)$ we mean the set of all one-to-one mappings from
$A$ to $B$, and  we write $S_n$ as a shorthand for $\Pi([n],[n])$.  We
assume familiarity with standard notions from computational complexity theory~\cite{pap:b:complexity}, parameterized
algorithmics~\cite{nie:b:invitation-fpt,cyg-fom-kow-lok-mar-pil-pil-sau:b:fpt},
and regarding approximation algorithms~\cite{vaz:b:approximation}.

Let $C$ be a (finite, nonempty) set of candidates. We refer to linear
orders over $C$ as \emph{preference orders} (or, \emph{votes}),
ranking the candidates from the most to the least appealing one.  We
write $\calL(C)$ to denote the set of all preference orders over $C$.
Every subset $\calD$ of $\calL(C)$ is called a \emph{domain} (of
preference orders over $C$) and, in particular, $\calL(C)$ itself is
the \emph{general domain}.  Later we will consider the single-peaked
and single-crossing domains.  

For
a vote $v$, we write $v \colon a \pref b$ to indicate that $v$ ranks
candidate $a$ higher than~$b$ (i.e., $v$ prefers~$a$ to~$b$). By
$\pos_v(c)$ we mean the position of candidate~$c$ in~$v$ (the
top-ranked candidate has position $1$, the next one has position $2$,
and so on).

An election $E = (C,V)$ consists of a set of candidates
$C = \{c_1,\ldots, c_m\}$ and a collection of voters
$V = (v_1, \ldots, v_n)$, where each voter $v_i$ has a preference
order, also denoted as $v_i$ (the exact meaning will always be clear
from the context and this convention will simplify and streamline our
discussions).  The preference orders 
always come from some domain $\calD$ (we use the general domain, unless stated
otherwise).

For a set $X$, a function $d \colon X \times X \rightarrow \mathbb{R}$
is a \emph{metric} if for each $x, y, z \in X$ it holds that
(i)~$d(x,y) \geq 0$, (ii)~$d(x,y) = 0$ if and only if $x = y$,
(iii)~$d(x,y) = d(y,x)$, and (iv) $d(x,z) \leq d(x,y) +
d(y,z)$. A~\emph{pseudometric} relaxes condition (ii) to the
requirement that $d(x,x) = 0$ for each $x \in X$. In particular, for a
pseudometric $d$ it is possible that $d(x,y) = 0$ when $x \neq y$.  We
use the terms pseudometric and distance interchangeably (however, we
also often use the term ``distance'' to refer to a value of a given
pseudometric between some two objects; we make sure that the meaning
is always clear).

We focus on the following three metrics between preference orders
(below, let $C$ be a set of candidates and let $u$ and $v$ be two
preference orders from $\calL(C)$):

\begin{description}
\item[Discrete Distance.]  The discrete distance between $u$ and $v$,
  $d_\discrete(u,v)$, is $0$ when $u$ and $v$ are identical and it
  is~$1$ otherwise.

\item[Swap Distance.] The swap distance between $u$ and $v$ (also
  known as the Kendall's Tau distance in statistics), denoted
  $d_\swap(u,v)$, is the smallest number of swaps of consecutive
  candidates that need to be performed within $u$ to transform it
  into~$v$.

\item[Spearman Distance.] The Spearman's distance (also known as the
  Spearman's footrule or the displacement distance) measures the total
  displacement of candidates in $u$ relative to their positions in
  $v$. Formally, we have:
\[
   d_\spearman(u,v) = \sum_{c\in C} |\text{pos}_v(c)-\text{pos}_u(c)|.
\]
\end{description}

We only consider such distances over preference orders that are
defined for all sets of candidates (as is the case for $d_\discrete$,
$d_\swap$, and $d_\spearman$). For a discussion of various other
distances between preference orders (also viewed as permutations), we
point to the overview of Deza and Deza~\cite{dez-dez:b:encyclopedia}.

Consider two sets of candidates, $C$ and $D$, of the same cardinality.
Let $\sigma$ be a bijection from $C$ to $D$.  We extend $\sigma$ to
act on preference orders $v \in \mathcal{L}(C)$ in a natural way,
so that $\sigma(v)\in \mathcal{L}(D)$ is the preference order where for each $c, c' \in C$, it holds that
$v \colon c \pref c' \iff \sigma(v) \colon \sigma(c) \pref
\sigma(c')$.

For an election $E = (C,V)$, where $V = (v_1, \ldots, v_n)$, a
candidate set $D$, and a bijection $\sigma$ from $C$ to $D$, by
$\sigma(E)$ we mean the election with candidate set $D$ and voter
collection $(\sigma(v_1), \ldots, \sigma(v_n))$.

\section{Election Isomorphism and Isomorphic Distances}\label{sec:ei}
In this section we define the notion of election isomorphism,
illustrate its usefulness, and introduce isomorphic distances.

Two elections are isomorphic if they are identical up to renaming the
candidates and reordering the voters. Formally, we have the following
definition.

\begin{definition}
  We say that elections $E=(C,V)$ and $E'=(C',V')$, where $|C| = |C'|$,
  $V = (v_1, \ldots, v_n)$, and $V' = (v'_1, \ldots, v'_n)$, are
  isomorphic if there is a bijection $\sigma\colon C \to C'$ and a
  permutation $\nu\in S_n$ such that $\sigma(v_i)=v'_{\nu(i)}$ for all
  $i\in [n]$.
\end{definition}

\begin{example}
  Consider elections $E = (C,V)$ and $E' = (C',V')$, such that
  $C = \{a,b,c\}$, $C' = \{x,y,z\}$, $V = (v_1,v_2,v_3)$,
  $V' = (v'_1, v'_2,v'_3)$, with preference orders:
\begin{align*}
 &  v_1 \colon a \pref b \pref c,\quad &&   v_2 \colon b \pref a \pref c, \quad && v_3 \colon c \pref a \pref b,\\
 &  v'_1 \colon y \pref x \pref z, \quad &&  v'_2 \colon x \pref y \pref z, \quad && v'_3 \colon z \pref x \pref y.
\end{align*}
$E$ and $E'$ are isomorphic, by mapping candidates $a$ to $x$, $b$ to
$y$, and $c$ to $z$, and voters $v_1$ to $v'_2$, $v_2$ to $v'_1$, and
$v_3$ to $v'_3$.
\end{example}
The idea of election isomorphism is quite natural and has already
appeared in the literature, though without using this name and usually
as a tool to achieve some specific goal.  For example,
E\u{g}ecio\u{g}lu and Giritligil~\cite{ege-gir:j:isomorphism-ianc}
refer to two isomorphic elections as members of the same
\emph{anonymous and neutral equivalence class (ANEC)} and study the
problem of sampling representatives of ANECs uniformly at random.
Hashemi and Endriss~\cite{has-end:c:diversity-indices} use the
election isomorphism idea in their analysis of preference diversity
indices.

While for many types of mathematical objects it is not clear if it is
possible to test their isomorphism (as witnessed by the famous
\textsc{Graph Isomorphism} problem; see, e.g., the report of Babai et
al.~\cite{bab-daw-sch-tor:j:graphi-isomorphism}), for elections this
task is surprisingly easy. Indeed, it suffices to consider
polynomially many bijections between the candidates, and for each
check if the elections become identical (after normalization).  Since
we will be interested in distances under which exactly the isomorphic
elections are at distance zero, this result indicates that such
distances are not doomed to be intractable.

\begin{proposition}\label{thm:ei-complexity}
  There is a polynomial-time algorithm that given two elections
  decides if they are isomorphic.
\end{proposition}
\begin{proof}
  Let $E = (C,V)$ and $E' = (C',V')$ be two input elections, where
  $C = \{c_1, \ldots, c_m\}$, $C' = \{c'_1, \ldots, c'_m\}$,
  $V = (v_1, \ldots, v_n)$ and $V = (v'_1, \ldots, v'_n)$.  Without
  loss of generality, let us assume that $v_1$'s preference order is
  $v_1 \colon c_1 \pref c_2 \pref \cdots \pref c_m.$ For each $v'_j$
  there is a bijection $\sigma_j$ such that $\sigma_j(v_1) =
  v'_j$. For each $\sigma_j$, we sort the votes in elections
  $\sigma_j(E)$ and $E'$ lexicographically (by assuming some arbitrary
  order over $C'$) and check if they are identical. If so, we accept.
  If we do not accept for any $\sigma_j$, then we reject.

  The algorithm runs in polynomial time because there are $n$
  $\sigma_j$'s to try.  Correctness follows from the fact that we need
  to map $v_1$ to some vote in $E'$ and we try all possibilities.
\end{proof}

In a follow-up work, Faliszewski et
al.~\cite{fal-sor-szu:c:subelection-isomorphism} have shown that the
problems of testing if two elections have isomorphic subelections of a
given size, or if a given election is isomorphic to a subelection of a
given, larger one, are $\np$-hard.

\subsection{Isomorphism of Preference Domains}

As an extended example of usefulness of the isomorphism idea,

we consider the single-peaked~\cite{bla:b:polsci:committees-elections}
and single-crossing~\cite{mir:j:single-crossing,rob:j:tax} domains.
For detailed discussions of both of them, see the surveys of Elkind et
al.~\cite{elk-lac-pet:t:domains,elk-lac-pet:c:restricted-domains}.

\begin{definition}
  Let $\calD \subseteq \calL(C)$ be a domain.
  \begin{enumerate}
  \item $\calD$ is single-peaked if there exists a linear order $>$
    over $C$ (referred to as the societal axis) such that for each $v$
    and each $j \in [|C|]$, the top $j$ candidates from $v$ form an
    interval according to the order $>$ (i.e., for each three
    candidates $a, b, c$, such that we have $a > b > c$, if $v$ ranks
    $a$ and $c$ among top $j$ candidates, then $v$ also ranks $b$
    among top $j$ candidates).

  \item $\calD$ is single-crossing if its preference orders can be
    ordered as $v_1, \ldots, v_n$ so that for each pair of candidates
    $a, b \in C$, as we consider $v_1, \ldots, v_n$ in this order, the
    relative ranking of~$a$ and~$b$ changes at most once.
  \end{enumerate}
\end{definition}

A single-peaked (single-crossing) domain is maximal if it is not
contained in any other single-peaked (single-crossing) domain. Each
maximal single-peaked domain $\calD \subseteq \calL(C)$ contains
$2^{|C|-1}$ preference orders (Monjardet~\cite{mon:survey}
attributes this fact to a 1962 work of Kreweras).

Since we can view a domain as an election that includes a single copy
of each preference order from the domain, our notion of isomorphism
directly translates to the case of domains, and we can formalize a
fundamental difference between the single-peaked and single-crossing
domains.

\begin{proposition}
  Each two maximal single-peaked domains over candidate sets of the
  same size are isomorphic. There are two maximal single-crossing
  domains over the same set of candidates that are not isomorphic.
\end{proposition}
\begin{proof}[Proof]
  For the first part of the proposition, it suffices to note that if
  $\calD$ and $\calD'$ are two maximal single-peaked domains (over
  candidate sets $\{x_1, \ldots, x_m\}$ and $\{y_1, \ldots, y_m\}$,
  respectively), with axes $>_1$ and $>_2$, such that: 
  \begin{align*}
   x_1 >_1 \cdots >_1 x_m && \text{and} && y_1 >_2 \cdots >_2 y_m,
  \end{align*}
  then a bijection that maps each $x_i$ to $y_i$ witnesses that the
  two domains are isomorphic.
  
  For the second part of the proposition, consider the following two
single-crossing domains of four candidates (each preference
  order is shown as a column, with the first candidate on top and the
  last one on the bottom):
$$
\begin{array}{ccccccc}
  a&b&b&b&b&d&d\\
  b&a&c&c&d&b&c\\
  c&c&a&d&c&c&b\\
  d&d&d&a&a&a&a
\end{array}
\qquad
\text{and}
\qquad
\begin{array}{ccccccc}
  d&d&d&a&a&a&a\\
  c&c&a&d&c&c&b\\
   b&a&c&c&d&b&c\\
   a&b&b&b&b&d&d   
\end{array}
$$
Both are maximal single-crossing domains and both are maximal Condorcet
domains~\cite{puppe2017condorcet}. 

They are not isomorphic as the first one has three different
candidates on top, and the second one has two.

\end{proof}

Given this result, it is very natural to ask, e.g., how many
nonisomorphic maximal single-crossing domains exist for a particular
number of candidates~$m$. Another interesting issue regards sampling
single-crossing domains uniformly at random (this task seems to be
nontrivial even if we do not require the domains to be nonisomorphic;
see the work of Szufa et al.~\cite{szu-fal-sko-sli-tal:c:map} for a
scenario where such sampling would be useful). We recommend both these
issues for future research.

\subsection{Isomorphic Distances}

Our next goal is to build distances that are invariant to renaming the
candidates and reordering the voters, and which ensure that two
elections are at distance zero exactly if they are isomorphic\footnote{In the
language of social choice theory, invariance to renaming the
candidates would be called \emph{neutrality} and invariance to reordering
the voters would be called \emph{anonymity}.}

(a weaker condition would be to require that isomorphic elections are
at distance zero, but distance zero does not need to imply
isomorphism; such distances also turned out to be
useful, e.g., in the map-of-elections framework~\cite{szu-fal-sko-sli-tal:c:map,boe-fal-nie-szu-was:c:metrics},
but we do not consider them in this paper).

Below we show how to extend distances over preference orders to the
type of distances that we are interested in.

\begin{definition}

  Let $d$ be a distance between preference orders.

  Let $E = (C,V)$ and $E' = (C',V')$ be two elections, where
  $|C| = |C'|$, $V = (v_1, \ldots, v_n)$ and
  $V' = (v'_1, \ldots, v'_n)$. We define the $d$-isomorphic distance
  between $E$ and $E'$ to be:

  \begin{equation*} 
    \id{d}(E,E') = \min_{\nu \in S_n}\min_{\sigma \in \Pi(C,C')}\sum_{i=1}^n  d(\sigma(v_i),v'_{\nu(i)})
  \end{equation*}
\end{definition}

We sometimes refer to the bijection $\sigma$ as the candidate matching
and to the permutation $\nu$ as the voter matching, and sometimes
instead of $\nu$, we use bijection $\tau \in \Pi(V,V')$, depending on
what is more convenient. Further, we will sometimes refer to $d_\swap$-ID
and $d_\spearman$-ID as the swap and Spearman isomorphic distances,
respectively.

Formally, the $d$-isomorphic distances are pseudometrics over the
space of elections with given numbers of candidates and
voters. Indeed, it is immediate to see that for every election
$E = (C,V)$ we have that $\id{d}(E,E) = 0$, and for each two elections
$E$ and $E'$ we have $\id{d}(E,E') = \id{d}(E',E)$. The triangle
inequality requires some care: We say that a distance $d$ over
preference orders is \emph{permutation-invariant} if for each two
candidate sets $C$ and $D$ of the same cardinality, each two
preference orders $u, v \in \calL(C)$, and each bijection $\sigma$
from $C$ to $D$, it holds that
$d(u,v) = d(\sigma(u),\sigma(v))$.

\begin{proposition}\label{prop:triangle}

  For each permutation-invariant distance $d$ over preference orders,
  $\id{d}$ satisfies the triangle inequality.

\end{proposition}
\begin{proof}
  Let $E = (C,V)$, $E' = (C',V')$, and $E'' = (C'',V'')$ be three
  elections such that $|C| = |C'| = |C''|$ and
  $V = (v_1, \ldots, v_n)$, $V' = (v'_1, \ldots, v'_n)$, and
  $V'' = (v''_1, \ldots, v''_n)$. Let $\sigma$ and $\sigma'$ be two
  bijections, from $C$ to $C'$ and from $C'$ to $C''$, respectively,
  and let $\nu$ and $\nu'$ be two permutations from $S_n$ such that:
  \begin{align*}
    \id{d}(E,E') &= \textstyle\sum_{i=1}^n  d(\sigma(v_i),v'_{\nu(i)}), \\  
    \id{d}(E',E'') & = \textstyle\sum_{i=1}^n  d(\sigma'(v'_i),v''_{\nu'(i)}).
  \end{align*}
  By definition of $\id{d}$, the fact that $d$ satisfies the triangle
  inequality, and by permutation invariance of $d$, we have that:
  \begin{align*}
     \id{d}(E,E'')  &\leq  \textstyle\sum_{i=1}^n  d(\sigma'(\sigma(v_i)),v''_{\nu'(\nu(i))}) \\ 
     &\leq \textstyle\sum_{i=1}^n  d(\sigma'(\sigma(v_i)),\sigma'(v'_{\nu(i)})) 
     +    \textstyle\sum_{i=1}^n  d(\sigma'(v'_{\nu(i)}),v''_{\nu'(\nu(i))}) \\
     &=    \textstyle\sum_{i=1}^n  d(\sigma(v_i),v'_{\nu(i)}) 
     +    \textstyle\sum_{i=1}^n  d(\sigma'(v'_i),v''_{\nu'(i)}) \\
     &= \id{d}(E,E') + \id{d}(E',E''),
  \end{align*}

  which gives our desired result.
\end{proof}

The $d$-isomorphic distances inherit some properties from the
underlying distances. For example, the Diaconis--Graham inequality~\cite{dia-gra:j:spearman}
says that for two preference orders $u$ and $v$ we have:
\begin{align}
  d_\swap(u, v) \leq d_\spearman(u, v) \leq 2 \cdot d_\swap(u,v)\label{ineq:diaconis-graham}
\end{align}
and, indeed, the same holds for the isomorphic distances.

\begin{proposition}\label{pro:one}
  Let $d_1$ and $d_2$ be two distances over preference orders for
  which there is some constant $c$, such that for every two preference
  orders $u$ and $v$ (over the same candidate set), it holds that
  $c \cdot d_1(u, v) \geq d_2(u, v)$.  Then, for every two elections
  $E$ and $E'$ (over candidate sets of the same cardinality and with
  the same numbers of voters) we have
  $c \cdot \id{d_1}(E,E') \geq \id{d_2}(E,E')$.
\end{proposition}

\begin{proof}

  Let $E = (C,V)$ and $E' = (C',V')$, where $V = (v_1, \ldots, v_n)$
  and $V' = (v'_1, \ldots, v'_n)$.  Let $\nu \in S_n$ and
  $\sigma \in \Pi(C,C')$ be the permutation of $[n]$ and the bijection
  from $C$ to $C'$ that give the value $\id{d_1}(E,E')$. We have that:
  \begin{align*}
     c \cdot \id{d_1}(E,E') = c \cdot \sum_{i=1}^n d_1(\sigma(v_i),v'_{\nu(i)}) 
                    \geq \sum_{i=1}^n d_2(\sigma(v_i),v'_{\nu(i)}) \geq \id{d_2}(E,E').
  \end{align*}
  This completes the proof.
\end{proof}

\begin{corollary}\label{cor:two-approx}
  For each two elections $E$ and $E'$ with the same numbers of
  candidates and voters we have:
\begin{align}
  \id{d_\swap}(E,E') \leq \id{d_\spearman}(E,E') \leq 2\cdot \id{d_\swap}(E,E').\label{ineq:spearman-2-swap}
\end{align}
\end{corollary}

In a follow-up work, Boehmer et
al.~\cite{boe-fal-nie-szu-was:c:metrics} analyzed various properties
of our distances (as well as of several other ones). In particular,
they established the maximum possible distance between two elections,
showed a few axiomatic properties, and experimentally evaluated
correlations between them.  One of the conclusions from their work is
that being able to compute swap and Spearman isomorphic distances
would be very useful.  Similar conclusions follow from the work of
Faliszewski et al.~\cite{fal-kac-sor-szu-was:c:microscope}, who have
shown connections between the swap isomorphic distance and the notions
of agreement, diversity, and polarization among the voters.
Unfortunately, as we will see in the following sections, not only are
swap and Spearman isomorphic distances intractable, but also the
typical workarounds---such as seeking approximate solutions or using
FPT algorithms---do not seem to be practical (consequently, the
authors of the just-cited papers focused on elections with up to $10$
candidates and up to $100$ voters and used brute-force algorithms).
Nonetheless, our FPT results do give hope that faster algorithms could
be developed later.

\section{Complexity of Computing Isomorphic Distances}

\begin{table*}

\centering
\begin{tabular}{c|ccc|ccc}
  \toprule
  &      & \sc with voter & \sc with candidate &  \multicolumn{3}{c}{parameter}\\
  $d$ & $d$-\textsc{ID} & \sc matching   & \sc matching    & $m$ & $n$ & $k$\\
  \midrule
  $d_\discrete $ & $\p$    & $\p$    & $\p$   & -- & -- & -- \\
  $d_\spearman $ & $\np$-complete & $\p$    & $\p$   & $\fpt$ & $\fpt$ & $\fpt$ \\
  $d_\swap $    & $\np$-complete & $\np$-complete & $\p$   & $\fpt$ & para-$\np$-hard & $\fpt$ \\
  \bottomrule
\end{tabular}

\caption{\label{tab:complexity}The complexity of computing isomorphic
  distances, also parametrized by the number of candidates $m$, the
  number of voters $n$, and the distance $k$.}
\end{table*}

To establish the complexity of computing our distances, we first need
to express this task as a decision problem (note that we overload the
terms $d$-\textsc{Isomorphic Distance} and $\id{d}$ to refer to both
the actual distance and the corresponding decision problem; we make
sure that it is always clear from the context which interpretation we
have in mind).

\begin{definition}
  Let $d$ be a distance over preference orders.  In the
  $d$-\textsc{Isomorphic Distance} problem (the $d$-{\sc ID} problem) we
  are given two elections, $E = (C,V)$ and $E' = (C',V')$ such that
  $|C| = |C'|$ and $|V| = |V'|$, and an integer $k$.  We ask if
  $\id{d}(E,E') \leq k$.
\end{definition}

We are also interested in two variants of this problem, the
\textsc{$d$-ID with Candidate Matching} problem, where the bijection
$\sigma$ between the candidate sets is given (and fixed), and the
\textsc{$d$-ID with Voter Matching} problem, where the voter
permutation $\nu$ is given (and fixed). They will help us later, when
discussing approximations and FPT algorithms. The former problem is in
$\p$ for polynomial-time computable distances, but this is not always
true for the latter.  We summarize our results in
Table~\ref{tab:complexity}.

\begin{proposition}\label{pro:matching}
  For a polynomial-time computable $d$, the problem \textsc{$d$-ID with Candidate
    Matching} is in $\p$.
\end{proposition}

\begin{proof}
  Let $E$ and $E'$ be our input elections and let $\sigma$ be the
  input matching between their candidates.  It suffices to compute a
  distance between every pair of votes (one from $\sigma(E)$ and the
  other from $E'$), build a corresponding bipartite graph (where
  vertices on the left are the voters from $\sigma(E)$, the vertices
  on the right are the voters from $E'$, all possible edges exist and
  are weighted by the distances between the votes they connect), and
  find the smallest-weight matching (see, e.g., the overview of Ahuja
  et al.~\cite{ahu-mag-orl:b:flows} for a classic, polynomial-time
  algorithm). The weight of the matching gives the value of the
  distance, and the matching itself gives the permutation
  $\nu$. 
\end{proof}

Using a similar argument as in the proof of
Propositions~\ref{thm:ei-complexity}, we show that $\id{d_\discrete}$
is in~$\p$.

\begin{proposition}\label{thm:aei-hamming}
  The $d_\discrete$-\textsc{ID} problem is in $\p$.
\end{proposition}

\begin{proof}
  Given two elections $E = (C,V)$ and $E' = (C',V')$, where
  $|C| = |C'|$, $V = (v_1, \ldots, v_n)$ and
  $V' = (v'_1, \ldots, v'_n)$, for each pair of votes $(v_i,v'_j)$ we
  construct a mapping $\sigma_{ij}\colon C\to C'$ such that
  $\sigma_{ij}(v_i) = v'_j$.

We  choose $\sigma_{ij}$ that leads to the smallest $\id{d_\discrete}$
  distance (we compute these distances using the
  \textsc{$d_\discrete$-ID with Candidate Matching} problem).

  The correctness of the algorithm follows from the observation that
  the largest possible value of $\id{d_\discrete}(E,E')$ is $n-1$; we
  can always ensure that at least one vote from $E$ matches perfectly
  a vote from $E'$. Thus one of the $\sigma_{ij}$ matchings must be
  optimal. 
\end{proof}

The elections for which the $\id{d_\discrete}$ distance is small are
nearly identical (up to renaming the candidates and reordering the
voters). In consequence, we do not expect such elections to frequently
appear in applications. Indeed, the same conclusion follows from the
work of Boehmer et al.~\cite{boe-fal-nie-szu-was:c:metrics}, who
analyze various (isomorphic) distances.  Thus, we need more
fine-grained distances, such as $\id{d_\swap}$ and $\id{d_\spearman}$.
Unfortunately, they are $\np$-hard to compute and, indeed, for
$\id{d_\swap}$ we inherit this result from the {\sc Kemeny Score} problem.

\begin{proposition}\label{proposition:swapnphard}
  The $d_\swap$-\textsc{ID} problem is $\np$-complete, even for elections
  with four voters.
\end{proposition}

\begin{proof}
  Membership in $\np$ is easy to see. To show hardness, we give a
  reduction from the \textsc{Kemeny Score} problem which is $\np$-complete~\cite{bar-tov-tri:j:who-won}.
  In this problem we
  are given an election $E = (C,V)$ and an integer $k$, and we ask if
  there exists a preference order $p$ over $C$ such that
  $\sum_{v \in V}d_\swap(v,p) \leq k$ (this value is known as the Kemeny score).
  We reduce it to
  the $d_\swap$-\textsc{ID} problem in a straightforward way: Given
  election $E = (C,V)$ and $k$, our reduction outputs election $E$, a
  new election $E' = (C',V')$ with the same number of candidates and
  voters as $E$, where all the votes are identical, and the integer
  $k$.

  The reduction runs in polynomial time.
  
  Let us now argue that it is correct.  Let $V = (v_1, \ldots, v_n)$
  and let $V'$ consist of $n$ copies of $v'$.  We note that
  $\id{d_\swap}(E,E') = \min_{\sigma \in \Pi(C,C')}\sum_{i=1}^n
  d_\swap(\sigma(v_i),v') = \min_{\sigma' \in \Pi(C',C)}\sum_{i=1}^n
  d_\swap(v_i,\sigma'(v'))$, which is at most $k$ if and only if there
  exists a preference order $p\in \calL(C)$ such that
  $\sum_{v \in V}d_\swap(v,p) \leq k$.
\end{proof}

The above reduction takes polynomial-time and preserves many natural parameters of a \textsc{Kemeny Score} instance such as the number of candidates, the number of voters etc.
Additionally, the $d_\swap$-isomorphic distance in a created $d_\swap$-\textsc{ID} instance is equal to the Kemeny score $k$.
Therefore, a hardness result for \textsc{Kemeny Score} may imply an analogous result for $d_\swap$-\textsc{ID} in a straightforward way.
In particular, \textsc{Kemeny Score} remains
$\np$-complete even for the case of four
voters~\cite{dwo-kum-nao-siv:c:rank-aggregation}, so having a matching between the voters cannot make the problem simpler
(this also follows from the fact that one of the elections in the
reduction consists of identical votes).

\begin{corollary}\label{cor:swap-voter-matching}
  \textsc{$d_\swap$-ID with Voter Matching} is $\np$-complete.
\end{corollary}

The situation for $d_\spearman$-{\sc ID} is somewhat different. In this case
Litvak's rule~\cite{litv:j:dist-cons}, defined analogously to
the Kemeny rule, but for the Spearman distance, is polynomial-time
computable~\cite{dwo-kum-nao-siv:c:rank-aggregation} and we can lift
this result to the case of \textsc{$d_\spearman$-ID with Voter
  Matching}. Without the voter matching, $d_\spearman$-{\sc ID} is
$\np$-complete.

\begin{proposition}\label{pro:spear-with-voters}
  \textsc{$d_\spearman$-ID with Voter Matching} is in $\p$.
\end{proposition}
\begin{proof}
  Let $E = (C,V)$ and $E' = (C',V')$ be two elections, where
  $|C| = |C'|$, $V = (v_1, \ldots, v_n)$ and
  $V' = (v'_1, \ldots, v'_n)$, and let $\nu \in S_n$ be the given
  voter matching. 
  For a bijection $\sigma\colon C \rightarrow C'$, the Spearman
  distance between $E$ and $E'$ is
  $\sum_{i=1}^n d_\spearman(\sigma(v_i),v'_{\nu(i)})$, which is:
  \begin{align*}
    \sum_{c' \in C'} \sum_{i=1}^n |\pos_{v_i}(\sigma^{-1}(c')) - \pos_{v'_{\nu(i)}}(c')|.
  \end{align*}
  In consequence, the cost induced by matching candidates $c \in C$
  and $c' \in C'$ is
  $\cost(c,c') = \sum_{i=1}^n |\pos_{v_i}(c) -
  \pos_{v'_{\nu(i)}}(c')|$. To solve our problem, it suffices to find
  a minimum cost perfect matching in a bipartite graph where
  candidates from $C$ are the vertices on the left, candidate from
  $C'$ are the vertices on the right, and for each $c \in C$ and
  $c' \in C'$ we have an edge from $c$ to $c'$ with cost
  $\cost(c,c')$.
\end{proof}

\begin{theorem}\label{theorem:spearmannphard}
  The $d_\spearman$-\textsc{ID} problem is $\np$-complete.
\end{theorem}

\begin{proof}

  The membership in $\np$ is clear and, hence, we focus on showing
  $\np$-hardness by giving a polynomial-time reduction from the
  classic \textsc{3SAT} problem. An instance of \textsc{3SAT} consists
  of a set of Boolean variables and a set of clauses, where each
  clause is a disjunction of exactly three literals (a literal is
  either a variable or its negation). We ask if there is an assignment
  of \textsc{True}/\textsc{False} values to the variables so that each
  clause contains at least one literal that evaluates to
  \textsc{True}. In such a case, we say that the input formula is
  satisfiable.

  \paragraph{Construction.}
  Let $I$ be an instance of \textsc{3SAT} with variables
  $x_1, \ldots, x_n$ and clauses $C_1, \ldots, C_m$.
  Our reduction
  forms an instance of the $d_\spearman$-\textsc{ID} problem by
  constructing two elections, $E = (C,V)$ and $E' = (C',V')$. We will
  use the following notation to simplify their description:
  \begin{align*}
    &P = 2(8m+1 + 6mn) \text{,} \\
    &T = 12n^2m + Pm\text{,~and} \\
    &L = (8m+1)(T + Pm)+ 6mn + 1 \text{.}
  \end{align*}
  The construction of the candidate sets $C$ and $C'$ is as follows:
  
  \begin{enumerate}
  \item For each variable $x_i$, $C$ contains candidates
    $x_i^{\mathrm{t}}$ and $x_i^{\mathrm{f}}$ and $C'$ contains
    candidates $y_i^{\mathrm{t}}$ and~$y_i^{\mathrm{f}}$.
    Intuitively, for each $i \in [n]$ candidates $x_i^{\mathrm{t}}$
    and $x_i^{\mathrm{f}}$ correspond to the variable $x_i$ and its
    negation, respectively, whereas candidates $y_i^{\mathrm{t}}$ and
    $y_i^{\mathrm{f}}$ correspond to their Boolean values (so a truth
    assignment where $x_i$ is \textsc{True} will correspond to a
    candidate matching where $x_i^{\mathrm{t}}$ is associated with
    $y_i^{\mathrm{t}}$ and $x_i^{\mathrm{f}}$ is associated with
    $x_i^{\mathrm{f}}$, whereas an assignment where $x_i$ is set to
    \textsc{False} will correspond to an opposite candidate matching;
    this intuition will become clearer in our correctness argument).

    Let
    $X = \{x_1^{\mathrm{t}}, x_1^{\mathrm{f}}, x_2^{\mathrm{t}},
    \ldots, x_n^{\mathrm{f}}\}$ and
    $Y = \{y_1^{\mathrm{t}}, y_1^{\mathrm{f}}, y_2^{\mathrm{t}},
    \ldots, y_n^{\mathrm{f}}\}$.

  \item We also add $T + Pm$ dummy candidates,
    $\{d_i\}_{i \in [T]} \cup \{g_{i, j}\}_{i \in [m], j \in [P]}$, to
    $C$, and $T + Pm$ dummy candidates,
    $\{d_i'\}_{i \in [T]} \cup \{g_{i, j}'\}_{i \in [m], j \in [P]}$,
    to~$C'$.  The role of the dummy candidates is to ensure that
    voters from appropriate groups are matched to each other in our
    elections.
    
    Let $D = \{d_1, d_2, \ldots, d_T\}$ and
    $D' = \{d_1', d_2', \ldots, d_T'\}$. For each $j \in [m]$, let
    $G_{j} = \{g_{j, 1}, \ldots, g_{j, P}\}$,
    $G_{j}' = \{g_{j, 1}', \ldots, g_{j, P}'\}$.    We set
    $G = G_1 \cup \cdots \cup G_m$, and
    $G' = G_1' \cup \cdots \cup G_m'$.
  \end{enumerate}
  For the sake of convenience, below we provide sizes of candidate sets that play
  particular role in our proof:
  \begin{align*}
    |X| = 2n, && |D| = T, && |G| = Pm 
    \text{\:\:(because for each $j \in [m]$ we have  $|G_j| = P$),}\\ 
    |Y| = 2n, && |D'| = T, && |G'| = Pm
    \text{\:\:(because for each $j \in [m]$ we have  $|G'_j| = P$).} 
  \end{align*}
  Consequently, $|C| = |C'| = 2n+T+Pm$.

  Let $\overrightarrow{C}$ denote the 
  following preference order:
  \begin{align*}
    x_1^{\mathrm{t}} &\succ x_1^{\mathrm{f}} \succ x_2^{\mathrm{t}} \succ x_2^{\mathrm{f}} \succ \ldots \succ x_n^{\mathrm{t}} \succ x_n^{\mathrm{f}}  \\
    & \succ g_{1, 1} \succ  g_{1, 2} \succ \ldots \succ g_{1, P} \succ g_{2, 1} \succ \ldots \succ g_{m, P} \succ 
     d_1 \succ d_2 \succ \ldots \succ d_T \text{.}
  \end{align*}
  For each $S \subseteq C$, by $\overrightarrow{S}$ we mean
  $\overrightarrow{C}$ restricted to the candidates from $S$.

  Preference order $\overrightarrow{C'}$, as well as its restrictions
  to subset of $C'$, are defined in the same way, but for the ``primed''
  candidates and for members of $Y$ instead of members of $X$.

  Next, we describe voter collections $V$ and $V'$. For each clause
  $C_j = \ell_{a} \lor \ell_{b} \lor \ell_{c}$, where $\ell_{a}$,
  $\ell_{b}$, and $\ell_{c}$ are literals, we include three voters in
  $V$ and three voters in $V'$. Let $x_a$, $x_b$ and $x_c$ be the
  variables corresponding to $\ell_{a}$, $\ell_{b}$, and $\ell_{c}$,
  respectively (so, for example, $\ell_a$ could be $x_a$, $\ell_b$
  could be $\neg x_b$, and $\ell_c$ could be $\neg x_c$).

  Without loss of generality, we assume that $a < b < c$.  We add the
  following clause voters (we put cardinalities of respective
  candidate groups on top of them):
  \begin{enumerate}
  \item We add two voters to $V$, each with the following preference
    order:    
    \begin{align*}
    \mathrm{free}(C_j) \colon \mybrace{\overrightarrow{X}}{2n} \succ \mybrace{\overrightarrow{G\setminus G_j}}{Pm-P} \succ \mybrace{\overrightarrow{D}}{T} \succ \mybrace{\overrightarrow{G_j}}{P}.
    \end{align*}
    By a small abuse of notation, we refer to both of these voters as
    the $\mathrm{free}(C_j)$ voters.  We also add a single voter
    $\mathrm{sat}(C_j)$, with preference order:
    \begin{align*}
      \mathrm{sat}(C_j) \colon \mybrace{\overrightarrow{X \setminus\{x_a^{\mathrm{t}}, x_b^{\mathrm{t}}, x_c^{\mathrm{t}}\}}}{2n-3} \succ \mybrace{\overrightarrow{G\setminus G_j}}{Pm-P} \succ \mybrace{\overrightarrow{D}}{T} \succ \mybrace{\overrightarrow{G_j}}{P} \succ \mybrace{x_a^{\mathrm{t}} \succ x_b^{\mathrm{t}} \succ x_c^{\mathrm{t}}}{3} \text{,}
    \end{align*}
    Intuitively, in our construction the $\mathrm{sat}(C_j)$ voter
    will correspond (through the voter matching) to a variable
    responsible for satisfying clause~$C_j$, whereas voters
    $\mathrm{free}(C_j)$ will correspond to those variables that, from
    the point of view of clause $C_j$, are free to have whatever
    values they like (or, in other words, whose values are not crucial
    for tha satisfaction of the clause).

  \item To describe the voters that we add to $V'$, for each
    $z \in \{a, b, c\}$ we define:
    \begin{align*}
      \mathrm{val}(\ell_z) =
      \begin{cases}
        y_z^{\mathrm{t}}       & \quad \text{if } \ell_z \text{ is a positive literal,}\\
        y_z^{\mathrm{f}}       & \quad \text{if } \ell_z \text{ is a negative literal.}\\
      \end{cases}
    \end{align*}
    Intuitively, $\mathrm{val}(\ell_z)$ corresponds to the value that
    variable $x_z$ should have to ensure that clause $C_j$ is
    satisfied.  We add three voters with the following preference
    orders to $V'$:
    \begin{alignat*}{9}
      &\mybrace{\overrightarrow{Y \setminus\{\mathrm{val}(\ell_a)\}}}{2n-1} &&\succ \mybrace{\overrightarrow{G\setminus G_j}}{Pm-P} &&\succ \mybrace{\overrightarrow{D}}{T} &&\succ \mybrace{\overrightarrow{G_j}}{P} &&\succ \mybrace{\mathrm{val}(\ell_a)}{1} \text{,} \\
      &\overrightarrow{Y \setminus\{\mathrm{val}(\ell_b)\}} &&\succ \overrightarrow{G\setminus G_j} &&\succ \;\overrightarrow{D} &&\succ \;\overrightarrow{G_j} &&\succ \mathrm{val}(\ell_b) \text{,~and} \\
      &\overrightarrow{Y \setminus\{\mathrm{val}(\ell_c)\}} &&\succ \overrightarrow{G\setminus G_j} &&\succ \;\overrightarrow{D} &&\succ \;\overrightarrow{G_j} &&\succ \mathrm{val}(\ell_c) \text{.}
    \end{alignat*}

  \end{enumerate}

  We also introduce a number of dummy voters.  Specifically, we add
  $L$ pairs of voters to $V$ where in each pair the voters have the
  following two preference orders:
  \begin{alignat*}{3}
    &\mybrace{x_1^{\mathrm{t}} \succ x_1^{\mathrm{f}} \succ x_2^{\mathrm{t}} \succ x_2^{\mathrm{f}} \succ \cdots \succ x_n^{\mathrm{f}}}{2n} &&\succ \mybrace{\overrightarrow{G}}{Pm+P} &&\succ \mybrace{\overrightarrow{D}}{T} \text{, and} \\
    & x_1^{\mathrm{f}} \succ x_1^{\mathrm{t}} \succ x_2^{\mathrm{f}} \succ x_2^{\mathrm{t}} \succ \cdots \succ x_n^{\mathrm{f}} &&\succ \;\;\;\overrightarrow{G} &&\succ \;\overrightarrow{D} \text{.}
  \end{alignat*}
  Next, we add $2L$ voters to $V'$, each with preference order:
  \begin{align*}
    &y_1^{\mathrm{t}} \succ y_1^{\mathrm{f}} \succ y_2^{\mathrm{t}} \succ  y_2^{\mathrm{f}} \succ\cdots \succ y_n^{\mathrm{f}}  \succ \overrightarrow{G} \succ \overrightarrow{D} \text{.}
  \end{align*}
  This completes the construction of elections $E$ and $E'$. Note that
  both of them have $(2n) + (T + Pm)$ candidates and $2L+3m$ voters. We
  ask if $d_\spearman$-\textsc{ID}$(E,E')$ is at most:
  \begin{align*}
      Z = 2Ln + (8m+1)(T + Pm) \text{.}
  \end{align*}

  \paragraph{Intuition.}
  Before we prove the correctness of the construction, let us give an
  intuitive explanation. The value $L$ is taken to be very large, in
  order to ensure that in the bijections witnessing the minimal
  Spearman distance, the $2L$ dummy voters from $V$ are matched to the
  $2L$ dummy voters from $V'$. Further, such a matching will ensure
  that
  \begin{inparaenum}[(i)]
  \item the dummy candidates are matched with the dummy candidates, and
 
  \item for each $i \in [n]$, each of the candidates
    $x_i^{\mathrm{t}}$ and $x_i^{\mathrm{f}}$ is matched either with
    $y_i^{\mathrm{t}}$ or with $y_i^{\mathrm{f}}$.
  \end{inparaenum}
  The dummy candidates are used to create a large block between
  particular candidates, so the optimal matchings must primarily take
  care of minimizing the number of situations where some matched
  candidates in the matched votes are on the other sides of this
  block. Each group $G_j$ is used to ensure that the voters
  corresponding to the same clause are matched to each other.

  \paragraph{Satisfiability Implies Bounded Distance.}
  Now, let us move to the formal proof of correctness. First, we show
  that if the input \textsc{3SAT} instance is satisfiable then the
  Spearman isomorphic distance between $E$ and $E'$ is at most $Z$.
  To this end, let $\Phi$ be a satisfying truth-assignment for the
  variables from $I$. We give 
  two bijections, between the candidates and the voters of our two
  elections, witnessing

  that $\id{d_\spearman}(E,E') \leq Z$:
  \begin{enumerate}
  \item Each dummy candidate $d_i \in C$ is mapped to $d_i' \in C'$,
    and each $g_{i, j} \in C$ is mapped to $g_{i, j}' \in C'$.
  \item For each variable $x_i$, if $\Phi(x_i) = \textsc{true}$ then
    $x_i^{\mathrm{t}}$ is mapped to $y_i^{\mathrm{t}}$ and
    $x_i^{\mathrm{f}}$ is mapped to $y_i^{\mathrm{f}}$. Otherwise,
    $x_i^{\mathrm{t}}$ is mapped to $y_i^{\mathrm{f}}$ and
    $x_i^{\mathrm{f}}$ to $y_i^{\mathrm{t}}$.
  \item Each dummy voter from $V$ is mapped to an arbitrary dummy voter
    from $V'$ (recall that the dummy voters from $V'$ have identical
    preferences, so it is irrelevant which dummy voters from $V$ are
    matched to which from $V'$).
  \item The matching of the clause voters is as follows. For each clause $C_j = \ell_{a} \vee \ell_{b} \vee \ell_{c}$,
    
    let $\ell_z$ be the literal that evaluates to \textsc{True} in
    $C_j$ (if there are several such literals, choose one
    arbitrarily). The voter from $V'$, corresponding to $C_j$ and
    whose preference order ends with $\mathrm{val}(\ell_z)$, is
    matched with $\mathrm{sat}(C_j)$. The remaining two clause voters
    from $V'$ that correspond to $C_j$ are matched with the two
    $\mathrm{free}(C_j)$ voters, arbitrarily.
  \end{enumerate}

  Let us compute the Spearman distance implied by this mapping. The
  dummy voters in $V$ and $V'$ agree on the positions of the dummy
  candidates. Further, since for each $i \in [n]$ each of the
  candidates $x_i^{\mathrm{t}}$ and $x_i^{\mathrm{f}}$ is mapped
  either with $y_i^{\mathrm{t}}$ or with $y_i^{\mathrm{f}}$, the dummy
  voters contribute a distance of $2Ln$ (every pair of dummy voters from $V$
  with different preference orders yields a distance of $2n$ to the two
  matched dummy voters from $V'$).

  Next, we analyze the clause voters. Consider a clause
  $C_j = \ell_{a} \vee \ell_{b} \vee \ell_{c}$. Each of the dummy
  candidates in the vote of $\mathrm{sat}(C_j)$ is shifted by two
  positions in comparison to the vote matched with
  $\mathrm{sat}(C_j)$. In each of the other two votes from $V'$ that
  correspond to $C_j$, the dummy candidates are shifted by one
  position in comparison to the votes matched with them (i.e., the
  $\mathrm{free}(C_j)$ voters). Thus, for the voters corresponding to
  $C_j$, the dummy candidates induce a distance of
  $4 (|G| + |D|) = 4(T + Pm)$. In vote $\mathrm{sat}(C_j)$, one of the
  candidates $x_a^{\mathrm{t}}$, $x_b^{\mathrm{t}}$, or
  $x_c^{\mathrm{t}}$ is on a position that is within distance of at
  most $2$ from the position of the matched candidate in the matched
  vote (because the corresponding variable is responsible for
  satisfying $C_j$). The remaining two out of these three candidates
  are within a distance of at most $T + Pm + 2n$ from the assigned
  candidates (indeed, this is the number of candidates in each of the
  elections). Each of the other nondummy candidates contribute a
  distance of at most $2n$. We can apply the same reasoning to the
  remaining voters corresponding to $C_j$ and we infer that the voters
  corresponding to $C_j$ induce at most the following distance (one
  way of thinking of the nondummy candidates part of the following
  equation is that for each of the three voters, each nondummy
  candidate contributes a distance of up to $2n$ but, additionally, one
  nondummy candidate in each of the $\mathrm{free}(C_j)$ votes and two
  nondummy voters from $\mathrm{sat}(C_j)$ also yield a distance of at
  most $T+Pm$):
  \begin{align*}
    &\underbrace{4(T + Pm)}_\text{from the dummy candidates} + \quad \underbrace{4(T + Pm) + 3\cdot 2n \cdot 2n}_\text{from the nondummy candidates} \quad = \quad
     8(T + Pm) + 12n^2 \text{.}
  \end{align*}
  Taking into account the dummy voters and all the $m$ clause voters,
  the upper bound on the distance is as follows (note that $12n^2m \leq T+Pm$):
  \begin{align*}
    &2Ln + 8m(T + Pm) + 12n^2m \leq 2Ln + (8m+1)(T + Pm) = Z \text{.}
  \end{align*}
  Consequently, our $d_\spearman$-\textsc{ID} instance has answer
  ``yes,'' as required. 

  \paragraph{Bounded Distance Implies Satisfiability.}
  Now, we consider the other direction, that is, we assume that
  $\id{d_\spearman}(E,E') \leq 2Ln + (8m+1)(T + Pm)$ and show that the
  input formula is satisfiable. For the sake of contradiction, let us
  assume the opposite, i.e., that the formula is not satisfiable.

  First, by the pigeonhole principle, we observe that some $2L-3m$
  dummy voters from $V$ must be matched with dummy voters from~$V'$
  (indeed, both elections have $2L+3m$ voters, of which $2L$ are
  dummy).  Among these voters, there are at least $L-3m$ pairs of
  dummy voters with different preference orders (recall the
  construction of the dummy voters in $V$). For each such pair, each
  candidate from $X$ that is not matched to a dummy candidate (i.e.,
  each candidate from $X$ that is matched to a member of $Y$)
  contributes at least one to the distance between the elections.  On
  the other hand, each candidate from $X$ that is matched to some
  dummy candidate contributes at least one unit of distance for every
  dummy voter that is matched to a dummy voter. Thus, if there are $s$
  candidates from $X$ that are matched to dummy candidates, then the
  distance is at least:
  \begin{align}
    \overbrace{(L-3m)(2n-s)}^{\substack{\text{members of $X$} \\ \text{matched to members of $Y$}}} &+ \overbrace{s \cdot (2L-3m)}^{\substack{\text{members of $X$ matched} \\ \text{to dummy candidates}}}  \label{eq:proof-np-h} \\
    & =2Ln -Ls -6mn +3ms +2Ls -3ms  \nonumber\\
    & =2Ln -6mn +Ls  \label{eq:proof-np-h-part} \\
    & =2Ln + (8m+1)(T+Pm) + 1 + (s-1)L  \nonumber\\
    & =Z + 1 + (s-1)L, \nonumber
  \end{align}
  which is greater than $Z$ for each $s \geq 1$. Thus, each candidate
  from $X$ must be matched to some candidate from $Y$.

  Similarly, we can show that each dummy candidate $d_i$ is mapped to
  $d_i'$, that each $g_{i, j}$ is mapped to $g_{i, j}'$, and that for
  each $i \in [n]$, each of the candidates $x_i^{\mathrm{t}}$ and
  $x_i^{\mathrm{f}}$ is matched either to $y_i^{\mathrm{t}}$ or to
  $y_i^{\mathrm{f}}$. For example, if some $d_i$ were not matched to
  $d'_i$, then the distance between the two elections would be at least:
  \[
    \underbrace{2n(L-3m)}_{\substack{\text{induced by the} \\ \text{nondummy candidates}}} + \underbrace{(2L-3m)}_{\substack{\text{induced by $d_i$ and $d'_i$ among} \\ \text{the matched dummy voters}}} = 2nL -6nm + 2L - 3m \geq 2nL - 6mn + L\text{, }
  \]
  which is equal to the value of \eqref{eq:proof-np-h-part} with $s=1$
  and, hence, is larger than $Z$. The same reasoning works for all the
  other dummy candidates. For the nondummy ones, i.e., for members of
  $X$, it works as well, but it requires a comment: As we have argued,
  members of $X$ contribute a distance of $2n(L-3m)$ due to the fact
  that each of them is ranked on two different positions (which differ
  by one) by the pairs of dummy voters with different preference
  orders (that are matched to the dummy voters in the other election,
  which all have equal preference orders); due to this, each of them
  certainly contributes a distance of one for each such pair. As there
  are $2n$ candidates in $X$ and there are at least $L-3m$ pairs of
  dummy voters with different preference orders that are matched to
  the dummy voters in the other election, in total this gives a
  distance of $2n(L-3m)$. If some $x_i^{\mathrm{t}}$ or
  $x_i^{\mathrm{f}}$ were matched to some $y_j^{\mathrm{t}}$ or
  $y_j^{\mathrm{f}}$, where $i \neq j$, then this matching would
  contribute an \emph{additional} distance of at least one, for each
  of the $2L-3m$ dummy voters matched to the other dummy voters.  This
  would contribute a distance of $2L-3m$.

  Next, we prove that each clause voter must be assigned to a clause
  voter associated with the same clause. Let us assume that this is
  not the case and that:
  \begin{enumerate}
  \item[] There is some clause voter $v$ from $V$, associated with
    clause $C_j$, that is either assigned to a dummy voter from $V'$
    or to some clause voter from $V'$ associated with a different
    clause.
  \end{enumerate}
  Consequently, by the pigeonhole principle: 
  \begin{enumerate}
  \item[] There is some clause voter from $V'$, associated with clause
    $C_j$, that is either assigned to a dummy voter from $V$ or to
    some clause voter from $V$ associated with a different clause.
  \end{enumerate}
  Regarding the first pair of voters, the candidates from $G_j$ yield
  a distance of at least $P(T-3)$, and they yield a distance at least
  $P(T-1)$ for the second one. The dummy voters that are assigned to
  other dummy voters yield a distance of at least $2n(L-3m)$.  Altogether,
  this gives at least:
  \begin{align*}
    2n(L-3m) + P(T-3) + P(T-1)
      &> 2n(L-3m) + PT  \\
      &= (2nL-6mn) + 2(8m+1 + 6mn)T  \\
      & > (2nL-6mn) + 6mn + (8m+1)2T \\
      & > 2nL + (8m+1)(T + Pm) = Z \text{,} 
  \end{align*}
  which contradicts our assumption that the elections are at most at
  distance $Z$.

  Finally, let us argue that the matching between the nondummy
  candidates corresponds to a satisfying truth assignment for the
  input \textsc{3SAT} formula.  By the preceding discussions, we know
  that the dummy voters yield a distance of $2nL$.  Next, let us
  consider some clause $C_j = \ell_{a} \lor \ell_{b} \lor
  \ell_{c}$. As we have argued before, within the votes corresponding
  to $C_j$, the dummy candidates induce a distance of
  $4 (|G| + |D|) = 4(T + Pm)$ (recall the calculations in the proof
  for the previous direction).  Further, for each of the two voters
  $\mathrm{free}(C_j)$ there is a nondummy candidate that induces a
  distance of at least $|G| + |D| = T + Pm$ (and one of them induces a
  distance larger by at least $1$). For the voter $\mathrm{sat}(C_j)$,
  there are two nondummy candidates that induce a distance of at least
  $|G| + |D| = T + Pm$, each.  In fact, if neither $x_a^{\mathrm{t}}$
  were matched to $\mathrm{val}(\ell_a)$, nor $x_b^{\mathrm{t}}$ were
  matched to $\mathrm{val}(\ell_b)$, nor $x_c^{\mathrm{t}}$ were
  matched to $\mathrm{val}(\ell_c)$, then there would be three such
  candidates and the distance between the two elections would be
  larger than:
  \begin{align*}
    2Ln + (8m+1)(T + Pm) =Z \text{, }
  \end{align*}
  which is impossible by our assumptions. This means that for each
  clause $C_j = \ell_{a} \vee \ell_{b} \vee \ell_{c}$ it must be that
  either $x_a^{\mathrm{t}}$ is matched to $\mathrm{val}(\ell_a)$ or
  $x_b^{\mathrm{t}}$ is matched to $\mathrm{val}(\ell_b)$ or
  $x_c^{\mathrm{t}}$ is matched to $\mathrm{val}(\ell_c)$.  However,
  this means that our input formula is satisfiable, via truth
  assignment that assigns to each variable $x_i$ value \textsc{True}
  exactly if candidate $x_i^{\mathrm{t}}$ is matched to
  $y^{\mathrm{t}}_i$. This contradicts the assumption that the input
  formula is not satisfiable and completes the proof.

\end{proof}

\section{Attempts to Circumvent Intractability}

In principle, the fact that $\id{d_\swap}$ and $\id{d_\spearman}$ are
intractable does not need to mean that they are difficult to solve in
practice. In this section we evaluate to what extent we can sidestep
our hardness results, either by considering approximation algorithms
or by designing fixed-parameter tractable (FPT) ones.

\subsection{Approximability}

Unfortunately, the news regarding approximation are negative.  It
turns out that we cannot hope for a constant-factor polynomial-time
approximation algorithm for our problems, unless \textsc{Graph
  Isomorphism} is in $\p$, which seems unlikely.  In the discussion
below, when we speak of $d_\swap$-\textsc{ID} and
$d_\spearman$-\textsc{ID} we mean the optimization variants of these
problems, where instead of deciding if the distance between two given
elections is at most a given value, we ask to compute the distance
between these elections.

\begin{theorem}\label{thm:hardness_approx_m}
  For each $\alpha < 1$, there is no polynomial-time
  $|C|^\alpha$-approximation algorithm neither for
  $d_\spearman$-\textsc{ID} nor for $d_\swap$-\textsc{ID}, unless the
  \textsc{Graph Isomorphism} problem is in $\p$.
\end{theorem}
\begin{proof}
  We first focus on the Spearman distance. Towards a contradiction,
  let us assume that there is a polynomial-time $|C|^\alpha$-approximation algorithm
  $\calA$ for $d_\spearman$-\textsc{ID} for some $\alpha < 1$.
  We will show that $\calA$ can
  be used to solve \textsc{Graph Isomorphism}.

  Let $I$ be an instance of \textsc{Graph Isomorphism}, consisting of
  two undirected graphs, $G = (\mathit{Vtx}, \mathit{Edg})$ and
  $G' = (\mathit{Vtx}', \mathit{Edg}')$. We use a somewhat nonstandard
  definition of \textsc{Graph Isomorphism} and assume that in $I$ we
  ask whether there are \emph{two} bijections,\footnote{Typically, one
    asks for a bijection between the vertices only, and it is assumed
    that the other bijection (between the edges) is induced by the
    first one. In our case it will be easier to work with two separate
    bijections. The two approaches yield the same decision problem.}
  $\sigma\colon \mathit{Vtx} \to \mathit{Vtx}'$ and
  $\tau\colon \mathit{Edg} \to \mathit{Edg}'$, such that for each
  $\mathit{edg} = \{\mathit{vtx}_1, \mathit{vtx}_2\} \in \mathit{Edg}$
  it holds that
  $\tau(\mathit{edg}) = \{\sigma(\mathit{vtx}_1),
  \sigma(\mathit{vtx}_2)\}$.  Without loss of generality, let us
  assume that $|\mathit{Vtx}| = |\mathit{Vtx}'| = n$ and
  $|\mathit{Edg}| = |\mathit{Edg}'| = m$. Further, we assume that
  there are no isolated vertices in any of the two graphs and that for
  each vertex there is an edge to which it is not incident (i.e., that
  neither of the graphs is a star).

  Given $I$, we construct an instance $J$ of $d_\spearman$-\textsc{ID}
  that consists of two elections, $E = (C, V)$ and $E' = (C', V')$, as
  follows. The voters in $V$ and $V'$ correspond to the edges from
  $\mathit{Edg}$ and $\mathit{Edg}'$, respectively. Further, for each
  vertex $\mathit{vtx} \in \mathit{Vtx}$ we add one candidate to $C$
  (denoted by the same symbol), and for each
  $\mathit{vtx}' \in \mathit{Vtx}'$ we add a candidate to
  $C'$. Additionally, we add
  $L = \lceil(2mn^2)^{\frac{1}{1-\alpha}}\rceil$ dummy candidates
  $D = \{d_1, \ldots, d_L\}$ to $C$ and $L$ dummy candidates
  $D' = \{d_1', \ldots, d_L'\}$ to $C'$.

  In consequence, our elections have $n+L$ candidates and $m$ voters
  each.

  Finally, we describe the
  preferences of the voters. For each edge
  $\mathit{edg} = \{\mathit{vtx}_1, \mathit{vtx}_2\}$, the voter in
  $E$ that corresponds to $\mathit{edg}$ has preference order of the form:
  \[
    \{\mathit{vtx}_1, \mathit{vtx}_2\} \succ d_1 \succ \cdots \succ d_L \succ \mathit{Vtx} \setminus \{\mathit{vtx}_1, \mathit{vtx}_2\}.
  \]
  That is, this voter ranks $\mathit{vtx}_1$ and $\mathit{vtx}_2$ on
  the first two positions (in some fixed, arbitrary order), then all
  the dummy candidates (in the order of increasing indices), and
  finally all remaining vertex candidates (in some fixed, arbitrary
  order).  We construct the preference orders of the voters from $V'$
  analogously.

  Now, observe that if there exists an isomorphism $(\sigma, \tau)$
  between the two graphs in the original instance, then the Spearman
  isomorphic distance between $E$ and $E'$ is at most $mn^2$. Indeed,
  it is apparent that this distance is witnessed by the bijection
  between voters $\sigma$ and by the function renaming the candidates
  $\tau$ combined with mapping $\tau_D$ such that $\tau_D(d_i) = d_i'$
  for each $d_i \in D$. Consequently, on this instance our
  approximation algorithm $\calA$ must find a distance that is lower
  than the following (note that $n < L$, so $1+\frac{n}{L} < 2$ and,
  as $\alpha < 1$, $(1+ \frac{n}{L})^\alpha < 2$):
  \begin{align*}
    mn^2 (L + n)^\alpha &= mn^2 \left(1 + \frac{n}{L}\right)^\alpha L^{\alpha} <  2mn^2 L^{\alpha} = 2mn^2 L^{\alpha-1} L \\
                        &=   \frac{2mn^2}{\lceil (2mn^2)^{\frac{1}{1-\alpha}} \rceil^{(1-\alpha)}} L \leq
                          \frac{2mn^2}{(2mn^2)^{\frac{1}{1-\alpha}(1-\alpha)}} L = L\text{.}
  \end{align*}
  Next, we assess the Spearman isomorphic distance in case when the answer to
  the original instance $I$ is ``no''. Let $\sigma$ and $\tau$ be two
  bijections between the sets of candidates and voters,
  respectively. First, let us analyze what happens when
  $\sigma(d_i) \notin D'$ for some $d_i \in D$. Since $d_i$ is ranked
  on the same position by each voter in $V$, and each vertex candidate
  appears at least once in the first two positions (because we do not
  have isolated vertices), and at least once in the last $n - 2$
  positions (because a graph is not a star), we infer that the
  Spearman isomorphic distance between $E$ and $E'$ is at least equal to
  $L$. Next, we move to the case when $\sigma(d_i) \in D'$ for each
  $d_i \in D$. Since $\sigma$ and $\tau$ (restricted to not include
  the dummy candidates) do not witness an isomorphism between $G$ and
  $G'$, it must be the case that there exists a voter $v$ and a
  candidate $c$ such that $c$ is ranked in the top two positions by
  $v$, yet $\sigma(v)$ ranks $\tau(c)$ in the last $n-2$
  positions. This means that the Spearman isomorphic distance is at
  least equal to~$L$.

  Thus, $\calA$ finds that the isomorphic distance between $E$ and
  $E'$ is lower than $L$ if and only if the answer to the original
  instance $I$ is ``yes''. This completes the proof for the case of
  the Spearman distance.

  The result for the swap distance follows by an almost identical
  construction and arguments. The difference is that now each of the
  elections we construct has $L+n$ dummy candidates (ranked in the
  same, consecutive blocks). We observe that if $G$ and $G'$ are
  isomorphic, then we still have that the swap isomorphic distance
  between the constructed elections is at most $nm^2$.  Thus, in this
  case, the approximation algorithm has to find a distance smaller
  than $L$. On the other hand, if the graphs are not isomorphic then
  the distance is at least $L$. To see this, note that if each dummy
  candidate is matched to a dummy candidate in the other election,
  then the same argument as for Spearman still holds (more precisely,
  in at least one vote, at least one vertex candidate would have to be
  swapped with all the dummy candidates).  Thus, let us consider the
  case where some dummy candidate $d_i$ is matched to a vertex
  candidate $\mathit{vtx}'_j$ in the other election. However, $d_i$ is
  ranked in the same position in all the votes in $E$, but
  $\mathit{vtx}'_j$ is ranked among top two candidates in at least one
  vote in $E'$, and is ranked among bottom $n-2$ candidates in at
  least one vote in $E'$. These two votes and $d_i$ contribute at
  least distance $L$ as, altogether, $d_i$ has to be swapped with at
  least $L$ out of $L+n$ dummy candidates (since at least $L$ of the
  dummy candidate are not matched to the vertex candidates).

\end{proof}

The same approach leads to an even stronger inapproximability result
with respect to the number of voters.

\begin{theorem}\label{thm:hardness_approx_n}
  For each $\alpha \geq 1$ there is no polynomial-time
  $|V|^\alpha$-approximation algorithm neither for
  $d_\spearman$-\textsc{ID} nor for $d_\swap$-\textsc{ID}, unless
  \textsc{Graph Isomorphism} is in $\p$.
\end{theorem}

\begin{proof}
  We use a very similar argument as in the proof of
  Theorem~\ref{thm:hardness_approx_m}, but this time we use the
  variant of \textsc{Graph Isomorphism} where the input graphs are
  regular (i.e., each vertex has the same degree).  This variant
  remains GI-complete~\cite{boo:j:isomorphism}, i.e., it is as
  computationally hard as the unrestricted one.
  For the sake of contradiction, let us assume that there is a
  polynomial-time $|V|^\alpha$-approximation algorithm $\calA$ for
  $d_\spearman$-\textsc{ID} for some $\alpha \geq 1$
  (without loss of generality, we assume
  that $\alpha$ is an integer). We use the same notation as in the
  proof of Theorem~\ref{thm:hardness_approx_m}.

  Let $G = (\mathit{Vtx}, \mathit{Edg})$ and
  $G' = (\mathit{Vtx}', \mathit{Edg}')$ be two input regular graphs,
  such that each vertex in each of them has the same degree $\delta$
  (if the graphs were regular but had vertices of different degrees,
  then they certainly would not be isomorphic). We also assume that
  $|\mathit{Vtx}| = |\mathit{Vtx}'| = n \geq 3$,
  $|\mathit{Edg}| = |\mathit{Edg}'| = m \geq 1$ (consequently, neither
  of the graphs is a star).  We form two elections, $E = (C, V)$ and
  $E' = (C', V')$. As opposed to the proof of
  Theorem~\ref{thm:hardness_approx_m}, this time the candidates will
  correspond to the edges and the voters will correspond to the
  vertices.  Specifically, for each edge
  $\mathit{edg} \in \mathit{Edg}$ we add one candidate to $C$ (denoted
  by the same symbol), and for each $\mathit{edg}' \in \mathit{Edg}'$
  we add one candidate to $C'$. Additionally, we add
  $L = {n^{\alpha+1}}m^2+1$ dummy candidates
  $D = \{d_1, \ldots, d_L\}$ to $C$ and $L$ dummy candidates
  $D' = \{d_1', \ldots, d_L'\}$ to $C'$.

  Our elections have $m+L$ candidates and $n$ voters each.
  
  For each vertex
  $\mathit{vtx}$, incident to edges $\mathit{edg}_1, \ldots,  \mathit{edg}_\delta$, we have a single voter
  with  preference order of the form:
  \[
    \{\mathit{edg}_1, \ldots, \mathit{edg}_\delta\} \succ d_1 \succ \cdots \succ d_L \succ \mathit{Edg} \setminus \{\mathit{edg}_1, \ldots, \mathit{edg}_\delta\}.
  \]
  The voters in $V'$ have preference orders constructed analogously.

  Using the same reasoning as in the proof of
  Theorem~\ref{thm:hardness_approx_m}, we observe that if $G$ and $G'$
  are not isomorphic then the Spearman isomorphic  distance between
  our two elections is at least $L$ (note that we needed our graphs to
  be regular to ensure that each dummy candidate is ranked on the same
  position in each vote, which is helpful in this argument). However,
  if $G$ and $G'$ are isomorphic, then the distance between $E$ and
  $E'$ is at most $nm^2$. Thus, algorithm $\calA$ on these two
  elections must output at most the following value (note that there
  are $n$ voters):
  \[
    (nm^2)n^\alpha = n^{\alpha+1}m^2 < L.
  \]
  Thus, using $\calA$, we can distinguish in polynomial time whether
  $G$ and $G'$ are isomorphic or not.

  To adapt the argument to the case of the swap distance, we proceed
  in the same way as in the proof of
  Theorem~\ref{thm:hardness_approx_m} (in this case, we change the
  number of dummy candidates from $L$ to $L+m$).
\end{proof}

Nonetheless, we can obtain \emph{some} approximation results. First,
we show that Theorem~\ref{thm:hardness_approx_m} is asymptotically
tight, by showing $O(|C|)$-approximation algorithms for our problems.

\begin{theorem}\label{thm:c-apx}
  There exists a polynomial time algorithm which outputs a solution that is
  $|C|$-approximate for $d_\spearman$-\textsc{ID} and
  $2|C|$-approximate for $d_\swap$-\textsc{ID}.
\end{theorem}
\begin{proof}
  We consider $d_\spearman$-\textsc{ID} first.  Let $E = (C,V)$ and
  $E' = (C',V')$ be our input elections,
  where 
  $|C| = |C'|$ and $|V| = |V'|$.  For each pair of voters $v \in V$
  and $v' \in V'$, let $\sigma_{v,v'}$ be the candidate matching such
  that $\sigma_{v,v'}(v) = v'$ (in other words, under this matching
  the distance between $v$ and $v'$ is zero).  For each matching
  $\sigma_{v,v'}$, the algorithm computes an optimal solution for
  \textsc{$d_\spearman$-ID with Candidate Matching} with this matching
  and outputs the one that yields the lowest distance.

  It is clear that this algorithm runs in polynomial time.

  Let us now analyze the approximation ratio that this algorithm
  achieves.  To this end, consider some optimal solution and let
  $\sigma$ and $\tau$ by the candidate and voter matchings used in
  this solution (in this proof, we use members of $\Pi(V,V')$ as voter
  matchings). Let $\opt$ be the distance between $E$ and $E'$, and let
  $v$ be a voter in $V$ that minimizes the value
  $d_\spearman(\sigma(v), \tau(v))$. Note that this value is at most
  ${\opt}/{|V|}$. Further, by definition, our algorithm
  outputs a distance that is at least as good as the one implied by
  $\sigma^* = \sigma_{v,\tau(v)}$.

  Next we analyze the distance between $E$ and $E'$ implied by
  $\sigma^*$. This matching ensures that the distance between $v$
  and $\tau(v)$ is zero, but

  in the optimal solution the distance between them can be up to
  $\opt/|V|$. This means that $\sigma^*$ differs from $\sigma$ on at
  most $\opt/|V|$ candidates (each candidate on which $\sigma$ differs
  from $\sigma^*$ contributes at least $1$ to the distance between $v$
  and $\tau(v)$ in the optimal solution).

  Consequently, the distance achieved by assuming candidate matching
  $\sigma^*$ is at most as follows (see explanations below):

  \begin{align}
    \opt + \frac{\opt}{|V|}  \cdot (|C|-1) \cdot (|V|-1) < |C| \cdot \opt,\label{ineq:c-apx}
  \end{align}
  This is so, because the distance induced by the candidates on which
  $\sigma^*$ and $\sigma$ agree is at most $\opt$, and the distance
  induced by each of the remaining $\opt/|V|$ candidates is at most
  $(|C|-1|)\cdot(|V|-1)$ (keep in mind that the distance between $v$
  and $\tau(v)$ is zero).

  This finishes the proof for the case of Spearman isomorphic 
  distance. 
  The result for the swap isomorphic distance follows by using the
  just-described algorithm and Corollary~\ref{cor:two-approx}.

\end{proof}

By sacrificing additional running time, we can slightly improve the
achieved approximation ratio. To do so, we need to find a part of the
solution using brute-force search.

\begin{proposition}
  For every constant $c$, there is a polynomial-time
  $(|C|-c)$-approximation algorithm for $d_\spearman$-\textsc{ID} and
  $2(|C|-c)$-approximation algorithm for $d_\swap$-\textsc{ID} (for
  the case where the input elections have more than $c$ candidates).
\end{proposition}
\begin{proof}
  We consider $d_\spearman$-\textsc{ID} first. Let $E = (C,V)$ and
  $E' = (C',V')$ be two input elections, with $|C|=|C'|$ and
  $|V| = |V'|$. Without loss of generality, we assume that $c$ is an
  integer and $|C| \geq 2c+2$ (otherwise we could solve the problem
  optimally by considering all candidate matchings and invoking
  Proposition~\ref{pro:matching}).  Based on the algorithm from
  Theorem~\ref{thm:c-apx}, it suffices to guess additional $2c+2$
  correct candidate-matches from an optimal solution and modify
  matching $\sigma^*$ accordingly. As a result, we produce a candidate
  matching that differs from the optimal one on at most
  $\opt/|V|-2c-2$ candidates.  W.l.o.g. we have
  $\opt/|V| - 2c - 2 > 0$; otherwise we guessed the optimal
  candidate-matching, hence we obtained the optimum solution due to
  Proposition~\ref{pro:matching}.  Therefore, the upper bound for the
  distance $\sol$ that our matching implies is as follows (analogously
  to Inequality~\ref{ineq:c-apx}, except that now we do not have
  guaranteed zero distance between the two voters used to define
  $\sigma^*$):
  \begin{align}
    \sol \leq &\opt + \left(\frac{\opt}{|V|} - (2c+2) \right)  \cdot (|C|-1) \cdot |V|\nonumber\\
    < &\opt + \left(\frac{\opt}{|V|} - (2c+2) \right)  \cdot |C| \cdot |V|\nonumber\\
    = &\opt \cdot (|C|+1) - (2c+2) \cdot |C| \cdot |V|.\label{ineq:c-c-apx}
  \end{align}
  
  In order to establish the claimed approximation ratio, we upper
  bound $\opt$ in terms of $|C| \cdot |V|$.  We note that any solution
  to $d_\spearman$-\textsc{ID} (hence, also an optimal one) yields
  distance at most $|C|^2 \cdot |V|$ (this is not a tight upper bound,
  but it suffices for our analysis; for a tight value, see the work of
  Boehmer et al.~\cite{boe-fal-nie-szu-was:c:metrics}).  Thus, if
  $\opt \geq \frac{|C|^2 \cdot |V|}{|C|-c}$ then any solution we
  provide is $(|C|-c)$-approximate.  Hence, w.l.o.g., we assume that
  $\opt < \frac{|C|^2 \cdot|V|}{|C|-c}$, which is equivalent to
  $\opt \cdot (1 - \frac{c}{|C|}) < |C| \cdot |V|$.  Then, we have:
  \begin{align*}
      \sol \stackrel{\eqref{ineq:c-c-apx}}{<} &\opt \cdot (|C|+1) - (2c+2) \cdot |C| \cdot |V|\\
      < &\opt \cdot (|C|+1) - (c+1) \cdot \opt \cdot \left(2 - \frac{2c}{|C|}\right)\\
      < &\opt \cdot (|C|+1) - (c+1) \cdot \opt = \opt \cdot (|C|-c).
  \end{align*}
  This finishes the proof for $d_\spearman$-\textsc{ID}.  For
  $d_\swap$-\textsc{ID}, the result follows by the same argument as in
  the proof of Theorem~\ref{thm:c-apx}.

  We conclude by noticing that guessing the modifications in the
  matching $\sigma^*$ increases the running time multiplicatively by a
  polynomial factor of $|C|^{4c+4} \cdot \poly(|C|,|V|)$. Thus the
  algorithm still runs in polynomial time.

\end{proof}

It is also possible to obtain algorithms whose approximation ratios
are a function of the number of voters. However, due to
Theorem~\ref{thm:hardness_approx_m}, this dependence has to be
superpolynomial. Below we show an example of such an algorithm.

\begin{proposition}
  There is a polynomial-time $|V|!$-approximation algorithm for
  $d_\spearman$-\textsc{ID} and $2|V|!$-approximation algorithm for
  $d_\swap$-\textsc{ID}.
\end{proposition}
\begin{proof}
  We start by considering $d_\spearman$-\textsc{ID}. Let $E = (C,V)$
  and $E' = (C',V')$ be the input elections, with $|C|=|C'|$ and
  $|V| = |V'|$. If $|C| \leq |V|!$, then we use the algorithm from
  Theorem~\ref{thm:c-apx}. Otherwise, if $|C| > |V|!$, then our input
  is sufficiently large to guess a voter matching in polynomial time
  and output an optimal solution, obtained via
  Proposition~\ref{pro:spear-with-voters}.

  To obtain the result for $d_\swap$-\textsc{ID}, we use the same
  algorithm as above and apply Corollary~\ref{cor:two-approx} (note
  that in this case we still use the polynomial-time algorithm for
  $d_\spearman$-\textsc{ID with Voter Matching} from
  Proposition~\ref{pro:spear-with-voters} as for $d_\swap$ the
  analogous problem is $\np$-complete; see
  Corollary~\ref{cor:swap-voter-matching}).
\end{proof}

In principle, one could improve the approximation ratio in the above
algorithm to $V!/c$ for a given constant $c$, but, as the algorithm
would not become particularly more useful, we omit such optimizations.
Instead, we leave it as an open problem if there are polynomial-time
approximation algorithms with considerably better approximation
ratios, such as $|V|^{\log|V|}$. While the existence of such
algorithms would mostly be of theoretical interest---as they are
unlikely to have practical value---they might lead to finding new
properties of the problem that could be of independent interest.
Another possible avenue for future research would be to improve our
inapproximability results from relying on the intractability of
\textsc{Graph Isomorphism} to using the $\p \neq \np$ assumption.
Alternatively, perhaps it is possible to find significantly stronger
approximation algorithms with mildly superpolynomial running time.

\subsection{Fixed Parameter Tractability}

As $\id{d_\swap}$ and $\id{d_\spearman}$ are both $\np$-complete and
hard to approximate, one further hope for their theoretical
tractability lays within parametrized complexity theory. We show that
there is a number of $\fpt$ algorithms for our problems, so in
principle this hope is not in vain.

To obtain an $\fpt$ algorithm for the parameterization by the number
of candidates,

it suffices to use simple brute-force algorithms (we guess the
matching between the candidates and invoke
Proposition~\ref{pro:matching}).

\begin{observation}\label{obs:fpt-m}
  Both $d_\swap$-{\sc ID} and $d_\spearman$-{\sc ID} are in $\fpt$ for
  parametrization by the number $m$ of candidates.
\end{observation}

For the parametrization by the number of voters, we know that
$\id{d_\swap}$ is para-$\np$-hard (i.e., it is $\np$-hard even for
some constant number of voters) due to
Proposition~\ref{proposition:swapnphard}, but for $\id{d_\spearman}$,
Proposition~\ref{pro:spear-with-voters} implies an $\fpt$ algorithm
(we guess the matching between the voters and invoke
Proposition~\ref{pro:spear-with-voters}).

\begin{observation}\label{obs:spearmanfptn}
  $d_\spearman$-{\sc ID} is in $\fpt$ for the parametrization by the 
  number of voters.
\end{observation}

Next, we consider the distance value $k$ as the parameter, which in the FPT
jargon would be referred to as the \emph{natural parameter}.

It turns out that both for the swap distance and the Spearman
distance, we have FPT algorithms.

\begin{proposition}
\label{FPT-Spearman}
$d_\spearman$-{\sc ID} is in $\fpt$ for the parametrization by the distance
value~$k$.
\end{proposition}

\begin{proof}

  Let $E=(C,V)$ and $E'=(C',V')$ be two elections with $|C| = |C'|$ and
  $|V_1|=|V_2|=n$.  If $k \geq n$ then the result follows from
  Observation~\ref{obs:spearmanfptn}, so we assume that $k < n$.

  If $\sigma\colon C\to C'$ and $\tau\colon V\to V'$ witness the
  fact that $\id{d_\spearman}(E_1,E_2)\le k$, then there is at least
  one voter $v$ from $V$ such that
  $d_\spearman(\sigma(v),\tau(v))=0$.

  Thus we guess a voter $v$ from $V$ and a voter $v'$ from $V'$ and
  compute the permutation $\sigma$ which makes votes $\sigma(v)$ and
  $v'$ identical. Then, we invoke Proposition~\ref{pro:matching} to
  test if this $\sigma$ indeed leads to distance at most $k$ between
  the elections (consequently, for $k < n$ the algorithm runs in
  polynomial time).
\end{proof}

\begin{theorem}\label{thm:dspear-k-fpt}
  $d_\swap$-{\sc ID} is in $\fpt$ for the parametrization by the distance value~$k$.
\end{theorem}

\begin{proof}
  Let $E = (C,V)$ and $E' = (C',V')$ be our input elections, where
  $|C| = |C'|$, $V = (v_1, \ldots, v_n)$, and
  $V' = (v'_1, \ldots, v'_n)$. Our goal is to check if
  $\id{d_\swap}(E,E')$ is at most $k$.  We treat the case where
  $k<n$ as in Proposition~\ref{FPT-Spearman} and,

  thus, we assume that $k \geq n$.

  We proceed by guessing a matching $\nu \in S_n$ between the voters
  that, supposedly, witnesses that $\id{d_\swap}(E,E')\le k$.  With
  respect to this matching, we say that a candidate $c \in C$ is
  \emph{happy} if there is a candidate $c' \in C'$ such that for every
  vote $v_i\in V$, we have that $c$ is ranked on the same position in
  $v_i$ as $c'$ is in $v_{\nu(i)}$ (i.e.,
  $\pos_{v_i}(c)=\pos_{v_{\nu(i)}}(c')$).  If a candidate is not
  happy, then we say that he or she is \emph{sad}.  Note that if the
  distance between $E$ and $E'$ is at most~$k$, then there are no more
  than $2k$ sad candidates (indeed, each swap of two neighboring
  candidates, in any of the votes from~$V$, can make at most two sad
  candidates happy, and all the candidates must be happy after at most
  $k$ swaps).

  Given $\nu$, our algorithm performs at most $k$ iterations, where in
  each iteration we execute the following steps:

  \begin{enumerate}
  \item We compute the set of candidates that are currently sad. We
    accept if there are no sad candidates and we reject if there is
    more than $2k$ of them.
  \item We guess a voter $v_i$ from $V$, a sad candidate $d \in C$, a
    candidate $c \in C$ such that $|\pos_{v_i}(c) - \pos_{v_i}(d)| \leq k$, and
    whether we swap $c$ with the candidate right before or right after
    him or her. Then we perform this swap of $c$ within $v_i$.
  \end{enumerate}
  We accept if there is a sequence of guesses such that after at most
  $k$ iterations there are no sad candidates, and we reject otherwise.
  Correctness follows from the fact that there is always a solution
  that starts by swapping a candidate that is within $k$ swaps of a
  sad one (we show this formally later).

  Fixed-parameter tractability follows from the fact that we perform
  $k$ iterations and in each iteration we have at most
  $n \cdot 2k \cdot (2k+1) \cdot 2$ possible guesses (corresponding to
  a voter, a sad candidate, distance of $c$ from this sad candidate,
  and the direction of the swap), which is at most $8k^3 + 4k^2$.
  Altogether, the algorithm's running
  time is $O^*(k! (8k^3+4k^2)^k) = O^*(9^kk!k^{3k}) = O^*(k^{5k})$,
  where the $O^*$ notation suppresses factors polynomial in the input size.
  The $k!$ factor comes from guessing the matching between the voters.

  It remains to argue that the algorithm is correct. If the algorithm
  accepts then certainly $\id{d_\swap}(E,E') \leq k$, so let us
  assume that $\id{d_\swap}(E,E') \leq k$ and show that the
  algorithm accepts. We note that if $\id{d_\swap}(E,E') \leq k$
  (and this happens for the matching $\nu$ among the voters that we
  guessed) then there exists some election $E^* = (C,V^*)$ with
  $V^* = (v_1^*, \ldots, v_n^*)$, such that:
  \begin{enumerate}
  \item $\sum_{i=1}^nd_\swap(v_i,v^*_i) \leq k$, and
  \item $\id{d_\swap}(E^*,E') = 0$ (and this is witnessed by $\nu$).
  \end{enumerate}
  In other words, $E^*$ is isomorphic to $E'$, but we have already
  renamed its candidates to be the same as in $E$, and we have
  reordered the voters according to $\nu$. Working with $E^*$ is
  equivalent to working with $E'$, but requires lighter notation.

  Now, consider a sequence $S$ of swaps of adjacent candidates that
  transforms $V$ into
  $V^*$.

  We claim that it must be possible to reorder the swaps in this
  sequence, so that the sequence can be executed (each swap regards
  adjacent candidates at the time when it is applied), but the first
  swap involves some voter $v_i \in V$ and a candidate whose position
  differs at most by $k$ from the position of some candidate $d$ that
  is sad in $E$.

  To this end, we note that there must be a vote $v_i \in V$ and a sad
  candidate $d$ (in $E$), for which $S$ includes a swap of $d$ within
  $v_i$ (otherwise the position of $d$ would not change after
  performing all the swaps and, so, $d$ would remain sad). Let $S_i$
  be the subsequence of swaps from $S$ that involve $v_i$.  Naturally,
  $S_i$ transforms $v_i$ into $v_i^*$ and, so, $S_i$ must contain
  exactly the swaps of candidates $c,\hat{c} \in C$ such that
  $v_i \colon c \pref \hat{c}$ and $v_i^* \colon \hat{c} \pref c$,
  listed in some order (this is a well-known property of the swap
  distance; see, e.g., the work of Elkind et
  al.~\cite{elk-fal-sli:c:swap-bribery} for a formal proof).  In
  particular, this means that if there are some candidates
  $c, \hat{c}$ such that in $v_i$, $c$ is right in front of $\hat{c}$
  but in $v_i^*$ they are ranked in reverse order (and not necessarily
  right next to each other) then there is an optimal sequence of swaps
  that transforms $v_i$ into $v_i^*$ and starts by swapping $c$ and
  $\hat{c}$. Now, for the sake of contradiction, let us assume that
  for each candidate $c$ such that
  $|\pos_{v_i}(c) - \pos_{v_i}(d)| \leq k$ and for each candidate
  $\hat{c}$ that is ranked right next to $c$, it is the case that
  their relative orders in $v_i$ and $v_i^*$ are the same (and, so, no
  optimal sequence of swaps that transforms $v_i$ into $v_i^*$ can
  swap them). Yet, $S_i$ includes a swap of $d$ and some other
  candidate $e \in C$.  By our assumption, we must make at least $k$
  swaps before~$d$ and~$e$ are ranked next to each other and their
  swap can be applied.  However, this means that $S_i$ (and,
  consequently, $S$) must include more than $k$ swaps, which is a
  contradiction. Thus, our assumption was false and there is a
  candidate $c$ such that $|\pos_{v_i}(c) - \pos_{v_i}(d)| \leq k$ and
  there is candidate $\hat{c}$ ranked right next to $c$, such that
  their relative orders in $v_i$ and $v_i^*$ are different. We can
  modify sequence $S_i$ to start by swapping $c$ and $\hat{c}$, and we
  can move the modified sequence $S_i$ to precede all other swaps in
  $S$.

  The above argument shows that the first iteration of our algorithm
  is correct. All the following iterations are correct be repeating
  the same reasoning. In the end, if $\id{d_\swap}(E,E')$ is at most
  $k$ then our algorithm will find matching $\nu$ and a sequence of
  swaps that witnesses this fact.

\end{proof}

While we have obtained several $\fpt$ algorithms for our problems,
their running times are too high for them to be practically useful. In
other words, our results should rather be seen as providing
a complexity-theoretic classification and not ready-to-implement
algorithms. Consequently, it is an interesting challenge to design
faster $\fpt$ algorithms. In particular, we ask if it is possible to construct an algorithm for $d_\swap$-{\sc ID} whose running time dependency on $k$ is lower than $O^*(k^{5k})$, as given in Theorem~\ref{thm:dspear-k-fpt}.
We note that there is an $O^*(2^{o(\sqrt{k})})$-time lower-bound for such an algorithm (assuming Exponential Time Hypothesis\footnote{Exponential Time Hypothesis is a popular conjecture on solving the satisfiability of propositional formulas in conjunctive normal form and it is used in parameterized complexity.
For a formal statement see, e.g., Conjecture 14.1 in~\cite{cyg-fom-kow-lok-mar-pil-pil-sau:b:fpt}.})
because of our parameter-preserving reduction in Proposition~\ref{proposition:swapnphard} and an $O^*(2^{o(\sqrt{k})})$-time lower-bound for {\sc Kemeny Score}~\cite[Theorem 18]{ArrighiFO020}.
In the case of {\sc Kemeny Score}, an $O^*(2^{O(\sqrt{k})})$-time algorithm is known~\cite{KarpinskiS10}.
It would be interesting to find analogous tight running-time bounds for our problems.
It is also natural to ask about existence of faster algorithms parameterized by the number of candidates $m$.
A simple brute-force algorithm for $d_\swap$-{\sc ID} runs in $O^*(m!)$-time (Observation~\ref{obs:fpt-m}) and there is an $O^*(2^{o(m)})$-time lower-bound (assuming Exponential Time Hypothesis) which follows from an analogous lower-bound for {\sc Kemeny Score}~\cite[Theorem 18]{ArrighiFO020}.
Note that for {\sc Kemeny Score}, an $O^*(2^m)$-time algorithm is known~\cite[Theorem 3]{BetzlerFGNR09}.

Another interesting direction would be to
consider structural parameters that characterize an instance in a more detailed way and which indeed may be small in practice.
Examples of such parameters include
a value $t$ such that for each pair of
matched voters (in an optimal solution), their distance is at most $t$,
or structural parameters studied for {\sc Kemeny Score}, such as the maximum range of positions that a candidate can have in any of the input votes in either of the input elections (see, e.g.,
the works of Betzler et al.~\cite{BetzlerFGNR09} and Arrighi et al.~\cite{ArrighiFO020,ArrighiFLO021}).

\section{Conclusions}

We have introduced isomorphic distances between elections and we have
studied the complexity of computing those based on the discrete, swap,
and Spearman distances between preference orders. Unfortunately, our
results turned out to mostly be negative: Isomorphic discrete distance
is computable in polynomial time, but does not appear to be useful (as
confirmed by Boehmer et al.~\cite{boe-fal-nie-szu-was:c:metrics}),
whereas swap isomorphic distance and Spearman isomorphic distance are
$\np$-complete and hard to approximate (assuming that \textsc{Graph
  Isomorphism} is not in $\p$), and their $\fpt$ algorithms are not
practical. Thus, so far, using optimized brute-force algorithms seems
to be the best way to compute these distances in practice, and that is
how they were computed in follow-up
works~\cite{boe-fal-nie-szu-was:c:metrics,fal-kac-sor-szu-was:c:microscope}.\footnote{One
  might also consider computing our distances using integer linear
  programming.  In a conference version of this paper we included a
  preliminary experiment in this direction, but found that this
  approach was worse than direct brute-force. In this version of the
  paper we omitted the experiment for the sake of focus.}
Nonetheless, our $\fpt$ algorithms give hope that there may be
stronger, mildly exponential-time algorithms. It might also be
fruitful to seek fast $\fpt$ approximation algorithms and consider
other parameterizations. Another possible direction for future
research would be to consider elections where the voters can provide
weak orders over the candidates.  However, likely, this would make
computing distances even more difficult as even deciding if two
election are isomorphic would be at least as hard as solving
\textsc{Graph Isomorphism} (indeed, each graph can be seen as an
election where vertices are the candidates, edges are the voters, and
each edge ranks the two vertices that it connects as tied on the first
place, and all the other vertices as tied on the second place).

More broadly, the idea of looking at distances that are invariant to
renaming candidates and voters, initiated in the conference version of
this paper, turned out to have quite some impact. In particular, it
has led to the development of the map-of-elections
framework~\cite{szu-fal-sko-sli-tal:c:map,boe-bre-fal-nie-szu:c:compass,boe-fal-nie-szu-was:c:metrics,szu-fal-jan-lac-sli-sor-tal:c:sampling-approval-elections,fal-kac-sor-szu-was:c:microscope},
where we gather a dataset of elections, compute the distances between
them using some metric, and then visualize them as points on a
plane. The key idea is to ensure that the distances between these
points resemble those from the metric. Since the elections regard
different candidates and voters, it is crucial for the metric to be
invariant to renaming the candidates and voters, as postulated in this
work. Our negative results have lead to considering distances that
ensure that isomorphic elections are at distance zero, but where the
opposite implication does not need to hold.

\subsection*{In Memory of Rolf Niedermeier}
Authors of this paper were friends and colleagues of Rolf
Niedermeier. Piotr Faliszewski was a friend and a frequent visitor to
his group in Berlin since 2013. He still works closely with many
researchers who set their first footsteps in science under Rolf's
supervision. In particular, Nimrod Talmon was Rolf's PhD student,
whereas Piotr Skowron was Piotr Faliszewski's PhD student, who joined
Rolf's group as a postdoc. Stanisław Szufa, also Piotr Faliszewski's
PhD student, was one of the last visitors to Rolf's group. We are all
very grateful to Rolf for his friendship, hospitality (and all the
visits to Berliner Philharmonie!), encouragement, and advice. Indeed,
this project is in part supported by Piotr Faliszewski's ERC project
PRAGMA that would not have happened without Rolf's counselling and
guidance. Rolf also provided invaluable advice to Piotr Skowron as he
prepared his ERC project PRO-DEMOCRATIC.

Finally, the current paper also has a special place in connection to our
collaboration with Rolf and his group: Map of elections, which stems
from this project, was largely developed between Piotr Faliszewski's
and Rolf's
groups~\cite{szu-fal-sko-sli-tal:c:map,boe-bre-fal-nie-szu:c:compass,boe-fal-nie-szu-was:c:metrics}.

\hyphenation{Piotr}
\subsection*{Acknowledgments}
When preparing the conference version of this paper,
Piotr Faliszewski was supported by AGH University statutory research
grant 11.11.230.337,

Piotr Skowron was supported by the Foundation for Polish Science
within the Homing programme (Project title: ``Normative Comparison of
Multiwinner Election Rules''), and

Arkadii Slinko was supported by Marsden Fund grant 3706352 of The
Royal Society of New Zealand.

Stanis{\l}aw Szufa was supported by NCN project
2018/\-29/N/\-ST6/01303. 

Later, this project has received funding from the European Research
Council (ERC) under the European Union’s Horizon 2020 research and
innovation programme (grant agreement No 101002854).  Krzysztof Sornat
was also partially supported by the SNSF Grant 200021\_200731/1.  Part
of this work was done while Krzysztof was a postdoc at AGH University,
Poland (during which time he was supported by the above ERC project).

We are very grateful to the anonymous AAAI and JCSS reviewers for their valuable comments.\bigskip

\begin{center}
  \includegraphics[width=3cm]{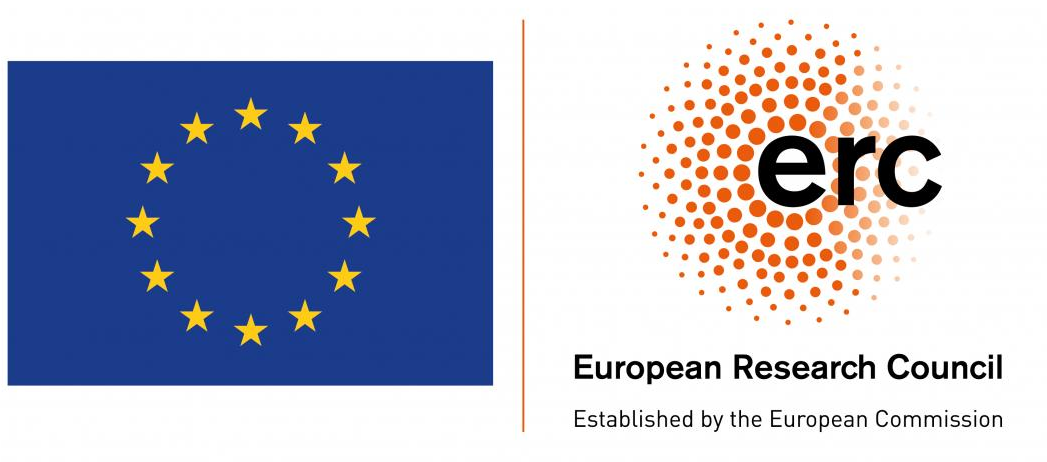}
\end{center}

\bibliographystyle{abbrv}
\bibliography{bib}

\begin{thebibliography}{10}

\bibitem{ahu-mag-orl:b:flows}
R.~Ahuja, T.~Magnanti, and J.~Orlin.
\newblock {\em Network Flows: Theory, Algorithms, and Applications}.
\newblock Prentice-Hall, 1993.

\bibitem{ArrighiFO020}
E.~Arrighi, H.~Fernau, M.~de~Oliveira~Oliveira, and P.~Wolf.
\newblock Width notions for ordering-related problems.
\newblock In {\em Proceedings of FSTTCS-2020}, pages 9:1--9:18, 2020.

\bibitem{ArrighiFLO021}
E.~Arrighi, H.~Fernau, D.~Lokshtanov, M.~de~Oliveira~Oliveira, and P.~Wolf.
\newblock Diversity in {K}emeny rank aggregation: {A} parameterized approach.
\newblock In {\em Proceedings of IJCAI-2021}, pages 10--16, 2021.

\bibitem{arv-koe-kur-vas:c:approximate-graph-isomorphism}
V.~Arvind, J.~K{\"{o}}bler, S.~Kuhnert, and Y.~Vasudev.
\newblock Approximate graph isomorphism.
\newblock In {\em Proceedings of MFCS-2012}, pages 100--111, 2012.

\bibitem{bab-daw-sch-tor:j:graphi-isomorphism}
L.~Babai, A.~Dawar, P.~Schweitzer, and J.~Tor{\'{a}}n.
\newblock The graph isomorphism problem ({Dagstuhl} {Seminar} 15511).
\newblock {\em Dagstuhl Reports}, 5(12):1--17, 2015.

\bibitem{bar-tov-tri:j:who-won}
J.~{{Bartholdi}}, III, C.~Tovey, and M.~Trick.
\newblock Voting schemes for which it can be difficult to tell who won the
  election.
\newblock {\em Social Choice and Welfare}, 6(2):157--165, 1989.

\bibitem{BetzlerFGNR09}
N.~Betzler, M.~R. Fellows, J.~Guo, R.~Niedermeier, and F.~A. Rosamond.
\newblock Fixed-parameter algorithms for {K}emeny rankings.
\newblock {\em Theoretical Computer Science}, 410(45):4554--4570, 2009.

\bibitem{bla:b:polsci:committees-elections}
D.~Black.
\newblock {\em The Theory of Committees and Elections}.
\newblock Cambridge University Press, 1958.

\bibitem{boe-bre-elk-fal-szu:c:frequency-matrices}
N.~Boehmer, R.~Bredereck, E.~Elkind, P.~Faliszewski, and S.~Szufa.
\newblock Expected frequency matrices of elections: {C}omputation, geometry,
  and preference learning.
\newblock In {\em Proceedings of NeurIPS-2022}, 2022.

\bibitem{boe-bre-fal-nie-szu:c:compass}
N.~Boehmer, R.~Bredereck, P.~Faliszewski, R.~Niedermeier, and S.~Szufa.
\newblock Putting a compass on the map of elections.
\newblock In {\em Proceedings of IJCAI-2021}, pages 59--65, 2021.

\bibitem{boe-fal-nie-szu-was:c:metrics}
N.~Boehmer, P.~Faliszewski, R.~Niedermeier, S.~Szufa, and T.~W\k{a}s.
\newblock Understanding distance measures among elections.
\newblock In {\em Proceedings of IJCAI-2022}, pages 102--108, 2022.

\bibitem{boo:j:isomorphism}
K.~Booth.
\newblock Isomorphism testing for graphs, semigroups, and finite automata are
  polynomially equivalent problems.
\newblock {\em SIAM Journal on Computing}, 7(3):273--279, 1978.

\bibitem{cyg-fom-kow-lok-mar-pil-pil-sau:b:fpt}
M.~Cygan, F.~Fomin, L.~Kowalik, D.~Lokshtanov, D.~Marx, M.~Pilipczuk,
  M.~Pilipczuk, and S.~Saurabh.
\newblock {\em Parameterized Algorithms}.
\newblock Springer, 2015.

\bibitem{dez-dez:b:encyclopedia}
M.~Deza and E.~Deza.
\newblock {\em Encyclopedia of Distances}.
\newblock Springer, 2009.

\bibitem{dia-gra:j:spearman}
P.~Diaconis and R.~Graham.
\newblock Spearman's footrule as a measure of disarray.
\newblock {\em Journal of the Royal Statistical Society. Series B},
  39(2):262--268, 1977.

\bibitem{dwo-kum-nao-siv:c:rank-aggregation}
C.~Dwork, R.~Kumar, M.~Naor, and D.~Sivakumar.
\newblock Rank aggregation methods for the {W}eb.
\newblock In {\em Proceedings of the 10th International World Wide Web
  Conference}, pages 613--622. ACM Press, Mar. 2001.

\bibitem{elk-fal-las-sko-sli-tal:c:2d-multiwinner}
E.~Elkind, P.~Faliszewski, J.~Laslier, P.~Skowron, A.~Slinko, and N.~Talmon.
\newblock What do multiwinner voting rules do? {An} experiment over the
  two-dimensional {E}uclidean domain.
\newblock In {\em Proceedings of AAAI-2017}, pages 494--501, 2017.

\bibitem{elk-fal-sli:c:swap-bribery}
E.~Elkind, P.~Faliszewski, and A.~Slinko.
\newblock Swap bribery.
\newblock In {\em Proceedings of the 2nd International Symposium on Algorithmic
  Game Theory}, pages 299--310. Springer-Verlag {\it Lecture Notes in Computer
  Science \#5814}, Oct. 2009.

\bibitem{elk-fal-sli:j:distance-rational}
E.~Elkind, P.~Faliszewski, and A.~Slinko.
\newblock Rationalizations of {C}ondorcet-consistent rules via distances of
  {H}amming type.
\newblock {\em Social Choice and Welfare}, 39(4):891--905, 2012.

\bibitem{elk-fal-sli:j:dr}
E.~Elkind, P.~Faliszewski, and A.~Slinko.
\newblock Distance rationalization of voting rules.
\newblock {\em Social Choice and Welfare}, 45(2):345--377, 2015.

\bibitem{elk-lac-pet:c:restricted-domains}
E.~Elkind, M.~Lackner, and D.~Peters.
\newblock Preference restrictions in computational social choice: Recent
  progress.
\newblock In {\em Proceedings of IJCAI-2016}, pages 4062--4065, 2016.

\bibitem{elk-lac-pet:t:domains}
E.~Elkind, M.~Lackner, and D.~Peters.
\newblock Preference restrictions in computational social choice: {A} survey.
\newblock Technical Report arXiv.2205.09092~[cs.GT], arXiv.org, 2022.

\bibitem{elk-sli:b:rationalization}
E.~Elkind and A.~Slinko.
\newblock Rationalizations of voting rules.
\newblock In F.~Brandt, V.~Conitzer, U.~Endriss, J.~Lang, and A.~D. Procaccia,
  editors, {\em Handbook of Computational Social Choice}, chapter~8, pages
  169--196. Cambridge University Press, 2016.

\bibitem{enelow1984spatial}
J.~M. Enelow and M.~J. Hinich.
\newblock {\em The Spatial Theory of Voting: An Introduction}.
\newblock Cambridge University Press, 1984.

\bibitem{enelow1990advances}
J.~M. Enelow and M.~J. Hinich.
\newblock {\em Advances in the Spatial Theory of Voting}.
\newblock Cambridge University Press, 1990.

\bibitem{ege-gir:j:isomorphism-ianc}
{\"{O}}.~E\u{g}ecio\u{g}lu and A.~Giritligil.
\newblock The impartial, anonymous, and neutral culture model: {A} probability
  model for sampling public preference structures.
\newblock {\em Journal of Mathematical Sociology}, 37(4):203--222, 2013.

\bibitem{fal-kac-sor-szu-was:c:microscope}
P.~Faliszewski, A.~Kaczmarczyk, K.~Sornat, S.~Szufa, and T.~W\k{a}s.
\newblock Diversity, agreement, and polarization in elections.
\newblock In {\em Proceedings of IJCAI-2023}, pages 2684--2692, 2023.

\bibitem{FaliszewskiSSST19}
P.~Faliszewski, P.~Skowron, A.~Slinko, S.~Szufa, and N.~Talmon.
\newblock How similar are two elections?
\newblock In {\em Proceedings of AAAI-2019}, pages 1909--1916, 2019.

\bibitem{fal-sor-szu:c:subelection-isomorphism}
P.~Faliszewski, K.~Sornat, and S.~Szufa.
\newblock The complexity of subelection isomorphism problems.
\newblock In {\em Proceedings of AAAI-2022}, pages 4991--4998, 2022.

\bibitem{gro-rat-woe:c:approximate-graph-isomorphism}
M.~Grohe, G.~Rattan, and G.~Woeginger.
\newblock Graph similarity and approximate isomorphism.
\newblock In {\em Proceedings of MFCS-2018}, pages 20:1--20:16, 2018.

\bibitem{has-end:c:diversity-indices}
V.~Hashemi and U.~Endriss.
\newblock Measuring diversity of preferences in a group.
\newblock In {\em Proceedings of ECAI-2014}, pages 423--428, 2014.

\bibitem{KarpinskiS10}
M.~Karpinski and W.~Schudy.
\newblock Faster algorithms for feedback arc set tournament, {K}emeny rank
  aggregation and betweenness tournament.
\newblock In {\em Proceedings of ISAAC-2010}, pages 3--14, 2010.

\bibitem{litv:j:dist-cons}
B.~Litvak.
\newblock Distances and consensus rankings.
\newblock {\em Cybernetics and Systems Analysis}, 19(1):71--81, 1983.
\newblock Translated from Kibernetika, No. 1, pp. 57--63, January--February,
  1983.

\bibitem{mat-wal:c:preflib}
N.~Mattei and T.~Walsh.
\newblock Preflib: A library for preferences.
\newblock In {\em Proceedings of ADT-2013}, pages 259--270, 2013.

\bibitem{mes-nur:b:distance-realizability}
T.~Meskanen and H.~Nurmi.
\newblock Closeness counts in social choice.
\newblock In M.~Braham and F.~Steffen, editors, {\em Power, Freedom, and
  Voting}. Springer-Verlag, 2008.

\bibitem{mir:j:single-crossing}
J.~Mirrlees.
\newblock An exploration in the theory of optimal income taxation.
\newblock {\em Review of Economic Studies}, 38:175--208, 1971.

\bibitem{mon:survey}
B.~Monjardet.
\newblock Acyclic domains of linear orders: A survey.
\newblock In S.~Brams, W.~Gehrlein, and F.~Roberts, editors, {\em The
  Mathematics of Preference, Choice and Order}, Studies in Choice and Welfare,
  pages 139--160. Springer Berlin Heidelberg, 2009.

\bibitem{nie:b:invitation-fpt}
R.~Niedermeier.
\newblock {\em Invitation to Fixed-Parameter Algorithms}.
\newblock Oxford University Press, 2006.

\bibitem{nit:j:closeness}
S.~Nitzan.
\newblock Some measures of closeness to unanimity and their implications.
\newblock {\em Theory and Decision}, 13(2):129--138, 1981.

\bibitem{pap:b:complexity}
C.~Papadimitriou.
\newblock {\em Computational Complexity}.
\newblock Addison-Wesley, 1994.

\bibitem{puppe2017condorcet}
C.~Puppe and A.~Slinko.
\newblock Condorcet domains, median graphs and the single-crossing property.
\newblock {\em Economic Theory}, pages 1--34, 2017.

\bibitem{red-vay-fla-cou:c:co-optimal-transport}
I.~Redko, T.~Vayer, R.~Flamary, and N.~Courty.
\newblock {CO}-optimal transport.
\newblock In {\em Proceedings of NeurIPS-2020}, 2020.

\bibitem{rob:j:tax}
K.~Roberts.
\newblock Voting over income tax schedules.
\newblock {\em Journal of Public Economics}, 8(3):329--340, 1977.

\bibitem{szu-fal-jan-lac-sli-sor-tal:c:sampling-approval-elections}
S.~Szufa, P.~Faliszewski, L.~Janeczko, M.~Lackner, A.~Slinko, K.~Sornat, and
  N.~Talmon.
\newblock How to sample approval elections?
\newblock In {\em Proceedings of IJCAI-2022}, pages 496--502, 2022.

\bibitem{szu-fal-sko-sli-tal:c:map}
S.~Szufa, P.~Faliszewski, P.~Skowron, A.~Slinko, and N.~Talmon.
\newblock Drawing a map of elections in the space of statistical cultures.
\newblock In {\em Proceedings of AAMAS-2020}, pages 1341--1349, 2020.

\bibitem{vaz:b:approximation}
V.~Vazirani.
\newblock {\em Approximation Algorithms}.
\newblock Springer, 2003.

\end{thebibliography}

\end{document}